\documentclass[10pt,parskip,a4paper,abstract=true,toc=bibliography,DIV=12]{scrartcl} %KOMA Script Article

\usepackage{amsmath,amssymb,amsthm}
\usepackage{mathtools}
\usepackage{interval} % open intervals ]a,b[
\usepackage{bbm} % blackboard bold
\usepackage{hyperref}
\newtheorem*{theorem*}{Theorem}
\newtheorem{theorem}{Theorem}[section]   
\newtheorem{proposition}[theorem]{Proposition}
\newtheorem{lemma}[theorem]{Lemma}        
\newtheorem{corollary}[theorem]{Corollary}

\newcommand{\N}{\mathbbm{N}}

\newcommand{\R}{\mathbbm{R}}
\newcommand{\C}{\mathbbm{C}}

\newcommand{\id}{\mathbbm{1}}            % Identitaet

\newcommand{\hilbert}{{\mathcal H}}      % Hilbert-Raum

\DeclareMathOperator{\im}{Im}            % Imaginaerteil
\DeclareMathOperator{\re}{Re}            % Realteil
\DeclareMathOperator{\sign}{sign}        % Signum
\DeclareMathOperator{\res}{res}          % Residuum
\DeclareMathOperator{\ran}{ran}          % Bildbereich, Range
\DeclareMathOperator{\tr}{tr}            % Trace, Spur

\DeclareMathOperator{\arcosh}{arcosh}

\numberwithin{equation}{section}

\begin{document}
\title{A Szeg\H{o} Limit Theorem Related to the Hilbert Matrix}
\author{
Peter Otte\thanks{peter.otte@rub.de}\\ 
Institute of Mathematics, University of Kassel, Germany
}
\date{\today}
%\keywords{Szeg\H{o} limit theorem, Hilbert matrix}
%\subjclass[msc2010]{Primary 81Q10, 34L40; Secondary 34L20, 34L25}
\maketitle
\begin{abstract}
The Szeg\H{o} limit theorem by Fedele and Gebert for matrices of the type identity minus Hankel matrix 
is proved for the special case $\id-\frac{\beta}{\pi}H_{N,\alpha}$ where $H_{N,\alpha}$ is the $N\times N$-Hilbert matrix, $\alpha\geq\frac{1}{2}$,
and $\beta\in\C$. The proof uses operator theoretic tools and a reduction to the classical Kac--Akhiezer theorem for the Carleman operator.
Thereby, the validity of the theorem for this special Hankel matrix can be extended from $|\beta|<1$ to $\beta\in\C\setminus\interval[open]{1}{\infty}$.
The bound on the correction term is improved to $O(1)$ instead of $o(\ln(N))$ for $\beta\in\C\setminus\interval[open right]{1}{\infty}$.
The limit case $\beta=1$ is derived directly from the asymptotics for general $\beta$.

\end{abstract}
%
%\tableofcontents
%%%%%%%%%%%%%%%%%%%%%%%%%%%%%%%%%%%%%%%%%%%%%%%%%%%%%%%%%%%%%%%%%%%%%%%%%%%%%%
%
\section{Introduction\label{intro}}
The Hilbert matrix appeared recently in the investigation of several problems such as 
Anderson's orthogonality catastrophe for Fermi gases \cite{GebertKuettlerMuellerOtte2016}, \cite{KnoerrOtteSpitzer2015}
and the spectral statistics of random matrices \cite{GebertPoplavskyi2019}.
In particular, all those problems led to some sort of Szeg\H{o} limit theorem for determinants. Subsequently, Fedele and Gebert \cite{FedeleGebert2019}
proved a Szeg\H{o} limit theorem for $\det(\id-\frac{\beta}{\pi} H_N)$ with a general $N\times N$ Hankel matrix $H_N$ and a parameter $\beta\in\C$, $|\beta|<1$.

Here, we give an alternative proof for the special case when $H_N$ is the Hilbert matrix. The proof uses operator theoretic methods. A key ingredient
is Wouk's integral formula \eqref{wouk} for the operator logarithm instead of the usual Taylor series. 
Thereby, the restriction $|\beta|<1$ can be replaced by the much weaker $\beta\notin\interval[open right]{1}{\infty}$
and the correction term is improved to $O(1)$ instead of $o(\ln(N)$ as in \cite{FedeleGebert2019}.
The limit case $\beta=1$ is directly deduced from the asymptotics for general $\beta$'s by use of a simple product formula, see \eqref{product}, which eventually
is a consequence of the third binomial formula.

To be more precise, we consider the general Hilbert matrix
\begin{equation*}
  H_{N,\alpha} = \Bigl( \frac{1}{j+k+\alpha} \Bigr)_{j,k=0,\ldots,N-1},\ N\in\N,\ \alpha>0  .
\end{equation*}
and obtain a Szeg\H{o} limit theorem for $\id-\frac{\beta}{\pi}H_{N,\alpha}$ with $\alpha\geq\frac{1}{2}$.
The case $0<\alpha<\frac{1}{2}$ is not treated herein since it would cause addtional technical difficulties.
The first main result of the paper is the following, see Theorem \ref{slt05t}. 

\begin{theorem}\label{intro01t}
Let $N\in\N$, $\alpha\geq\frac{1}{2}$ and $\beta\in\C\setminus\interval[open right]{1}{\infty}$. Then, 
the Hilbert matrix $H_{N,\alpha}$ satisfies
\begin{equation*}
  \det\bigl(\id - \frac{\beta}{\pi} H_{N,\alpha}\bigr) = \exp[ 2n_\alpha(N) \gamma(\beta) + O(1) ]\ \text{as}\ N\to\infty,\
   n_{\frac{\alpha}{2}}(N) = \frac{1}{4} \ln\bigl(\frac{N+\frac{\alpha}{2}}{\frac{\alpha}{2}}\bigr),
\end{equation*}
with the coefficient
\begin{equation*}
  \gamma(\beta) = \frac{1}{\pi^2} [\arcosh(-\beta)]^2 + \frac{1}{4} .
\end{equation*}
\end{theorem}

The use of $n_{\frac{\alpha}{2}}(N)$ instead of the simple logarithm $\ln(N)$ arises naturally during the proof, see below. Obviously,
\begin{equation*}
   n_{\frac{\alpha}{2}}(N) \sim \frac{1}{2}\ln(N) .
\end{equation*}
Note that $\gamma(\beta)$ is given in a different but equivalent form in \cite{FedeleGebert2019}, see \eqref{fedele_gebert}.

The proof of the theorem consists of two parts. In the first part, we relate the Hilbert matrix $H_{N,\alpha}$ to an integral operator $G_{N,\alpha}$ such that,
Lemma \ref{hm01t},
\begin{equation*}
  \det(\id - \frac{\beta}{\pi}H_{N,\alpha}) = \det(\id - \frac{\beta}{\pi}G_{N,\alpha}) .
\end{equation*}
The idea here is, essentially, to write the matrix entries of $H_{N,\alpha}$ as Laplace transforms
\begin{equation*}
  \frac{1}{j+\alpha} = \int_0^\infty e^{-(j+\alpha)x} \, dx .
\end{equation*}
We then show, Proposition \ref{hm05t}, that
\begin{equation*}
   \det(\id - \frac{\beta}{\pi}G_{N,\alpha}) \bigr) 
    = \det(\id - \frac{\beta}{\pi} P_{[\frac{\alpha}{2},N+\frac{\alpha}{2}]}KP_{[\frac{\alpha}{2},N+\frac{\alpha}{2}]}) \times \Delta(\beta)
\end{equation*}
Here $P_{[a,b]}$ denotes the orthogonal projection corresponding to the characteristic function $\chi_{[a,b]}$ of the interval $[a,b]$
and $K$ is the Carleman operator
\begin{equation*}
  (K\varphi)(x) = \int_0^\infty \frac{1}{x+y} \varphi(y)\, dy .
\end{equation*}
The so-called perturbation determinant $\Delta(\beta)$, cf. \eqref{perturbation_determinant}, can be shown to satisfy 
\begin{equation*}
  \ln(\Delta(\beta)) = O(1)\ \text{as}\  N\to\infty.
\end{equation*}
Here is where Wouk's integral formula \eqref{wouk} is used, see \eqref{determinant_integral}.

In the second part, we transform the Carleman operator $K$ unitarily to a convolution operator $K_0$, Lemma \ref{slt02t}. Since the projection
$P_{[\frac{\alpha}{2},N+\frac{\alpha}{2}]}$ has to be transformed accordingly $N$ becomes $n_{\frac{\alpha}{2}}(N)$.
Finally, we apply a general version of the classical Kac--Akhiezer theorem, Proposition \ref{slt01t}, to $K_0$ thereby completing the proof

If we had used the Taylor series of the logarithm as in \cite{FedeleGebert2019} we would have to work with
\begin{equation*}
  \ln\Bigl( \det\bigl(\id-\frac{\beta}{\pi}H_{N,\alpha}\bigr) \Bigr) = -\sum_{n=1}^\infty \frac{1}{n} \frac{\beta^n}{\pi^n} \tr(H_{N,\alpha}^n) .
\end{equation*}
However, the infinite series restricts the result to those $\beta$ for which the series converges, to wit $|\beta|<1$.

The second main result concerns the limit case $\beta=1$, see \ref{lc07t}.

\begin{theorem}\label{intro02t}
Let $\alpha\geq\frac{1}{2}$. Then,
\begin{equation*}
  \ln\bigl(\det(\id - \frac{1}{\pi}H_{N,\alpha})\bigr)
     = 2n_{\frac{\alpha}{2}}(N)  \gamma(1) + o(\ln(N)) \ \text{as}\ N\to\infty,\
        n_{\frac{\alpha}{2}}(N)=\frac{1}{4}\ln\Bigl( \frac{N+\frac{\alpha}{2}}{\frac{\alpha}{2}} \Bigr) 
\end{equation*}
with $\gamma(1)=-\frac{4}{3}$.
\end{theorem}

The key idea of the proof is to write, Lemma \ref{det01t},
\begin{equation*}
  \frac{1}{\det(\id-H_{N,\alpha})} = \prod_{m=0}^\infty \det\bigl( \id + (\frac{1}{\pi} H_{N,\alpha} )^{2^m} \bigr)
\end{equation*}
and use, at least formally, the asymptotics of each factor from Corollary \ref{slt06t}. The corollary itself follows easily from Theorem \ref{slt05t}
with the aid of the roots of unity. This idea can be made rigorous yielding, however, only a lower bound for the desired asymptotics, Proposition \ref{lc07t}.
Fortunately, since $H_{N,\alpha}$ is non-negative operator an upper bound,Proposition \ref{lc00t}, follows immediately from
\begin{equation*}
  \det(\id - \frac{1}{\pi}H_{N,\alpha}) \leq \det(\id - \frac{\beta}{\pi}H_{N,\alpha}),\ \beta<1,
\end{equation*}
and Theorem \ref{slt05t}.

The limit case $\beta=1$ was (for a special $\alpha$) also treated in \cite[Thm. 1.4]{GebertPoplavskyi2019}.
The method used therein relied on the explicit diagonalization of the infinite Hilbert matrix.

\section{Determinants\label{det}}
For a trace class operator $A:\hilbert\to\hilbert$ one can define the determinant $\det(\id-A)$. One way to do this is via the trace
\begin{equation}\label{determinant}
  \det(\id-A) = \tr[\ln(\id-A)] .
\end{equation}
In order to define the operator logarithm we recall the formula for the principal branch of the logarithm
\begin{equation}\label{logarithm}
  \ln(1-z) = -z\int_0^1 \frac{1}{1-rz},\ z\in \C\setminus\interval[open right]{1}{\infty} .
\end{equation}
This generalizes to Wouk's integral formula \cite{Wouk1965}
\begin{equation}\label{wouk}
  \ln(\id-A) = -\int_0^1 A(\id-rA)^{-1}\, dr
\end{equation}
which is valid whenever $\sigma(A)\cap\interval[open right]{1}{\infty}=\emptyset$. 
For alternative definitions and further properties
see e.g. \cite[XIII]{ReedSimon1978}. Standard estimates are (see \cite[Lemma 4, p. 323]{ReedSimon1978})
\begin{gather}
  |\det(\id-A)| \leq e^{\|A\|_1},  \label{det_estimate1}\\
  |\det(\id-A)-\det(\id-B)| \leq \|A-B\|_1 \exp\bigl[ \|A\|_1 + \|B\|_1 + 1\bigr] \label{det_estimate2}
\end{gather}
Another estimate, which is of special importance herein (see Section \ref{lc}), is based upon the infinite product
\begin{equation}\label{product}
  \frac{1}{1-x} = \prod_{m=0}^\infty (1 + x^{2^m}),\ x\in\R,\ |x|<1,
\end{equation}
more precisely on the version for determinant.

\begin{lemma}\label{det01t}
Let $A:\hilbert\to\hilbert$ be a trace class operator with $\|A\|<1$. Then,
\begin{equation}\label{det01t01}
  \frac{1}{\det(\id-A)} = \prod_{m=0}^\infty \det(\id+A^{2^m})
\end{equation}
where the infinite product converges absolutely. Furthermore,
\begin{equation}\label{det01t02}
  \frac{1}{\det(\id-A)} \leq \prod_{m=0}^M \det(\id + A^{2^m}) \exp\biggl[ \sum_{m=M+1}^\infty \|A^{2^m}\|_1\biggr],\ M\in\N_0 .
\end{equation}
\end{lemma}
\begin{proof}
We start off from the analogon of \eqref{product}
\begin{equation*}
  \frac{1}{\det(\id-A)} = \frac{1}{\det(\id-A^{2^M})} \prod_{m=0}^{M-1}\det(\id+A^{2^m}),\ M\in\N .
\end{equation*}
By \eqref{det_estimate1}, \eqref{det_estimate2}, and H\"older's inequality for the trace norm
\begin{equation*}
  \Bigl|\prod_{m=0}^{M-1}\det(\id+A^{2^m})\Bigr| \leq \prod_{m=0}^{M-1} (1+\|A\|^{2^m-1} \|A\|_1)
\end{equation*}
and
\begin{equation*}
  |\det(\id-A^{2^M}) -1 | \leq \|A\|^{2^M-1} \|A\|_1 \exp \bigl[ \|A^{2^M-1}\| \|A\|_1 + 1\bigr]   .
\end{equation*}
Using the assumption $\|A\|<1$ we deduce
\begin{equation*}
  \frac{1}{\det(\id-A)} 
    = \lim_{M\to\infty} \frac{1}{\det(\id-A^{2^M})} \prod_{m=0}^{M-1}\det(\id+A^{2^m})
    =  \prod_{m=0}^\infty\det(\id+A^{2^m}) .
\end{equation*}
This is \eqref{det01t01}. Finally, write
\begin{equation*}
  \frac{1}{\det(\id-A)}= \prod_{m=0}^M \det(\id+A^{2^m}) \times\prod_{m=M+1}^\infty \det(\id+A^{2^m})
\end{equation*}
and apply \eqref{det_estimate1} to the second factor. This shows \eqref{det01t02}.
\end{proof}

The determinants of two operators $A$ and $B$ are related via the perturbation determinant $\Delta(A,B)$
\begin{equation}\label{perturbation_determinant}
  \det(\id-A) = \det(\id-B) \times \Delta(A,B),\ \Delta(A,B) \coloneqq \det(\id-(\id-A)^{-1}(B-A)) .
\end{equation}
Wouk's formula \eqref{wouk} yields
\begin{equation}\label{determinant_integral}
  \ln(\Delta(A,B)) = - \tr\bigl[ (B-A)\int_0^1 (\id - rA - (1-r)B)^{-1}\, dr \bigr] .
\end{equation}

\section{Hilbert matrix and Carleman operator\label{hm}}
The Hilbert matrix is
\begin{equation}\label{hilbert_matrix}
  H_{N,\alpha} = \Bigl( \frac{1}{j+k+\alpha} \Bigr)_{j,k=0,\ldots,N-1},\ \alpha>0 .
\end{equation}
It is well-known that as an operator $H_{N,\alpha}: \C^N\to \C^N$ it satisfies
\begin{equation}\label{hilbert_matrix_spectrum}
  0 \leq H_{N,\alpha} \ \text{for}\ \alpha>0\ \text{and}\ H_{N,\alpha} < \pi\id\ \text{for}\ \alpha\geq \frac{1}{2}
\end{equation}
in the sense of quadratic forms. We will not treat the case $0<\alpha<\frac{1}{2}$ and, thus, do not need the corresponding norm.
With the aid of the Laplace transform
\begin{equation*}
  \frac{1}{j+\alpha} = \int_0^\infty e^{-jx} e^{-\alpha x} \, dx ,\ \alpha > 0
\end{equation*}
we obtain a Hankel integral operator with, essentially, the same spectrum as $H_{N,\alpha}$.

\begin{lemma}\label{hm01t}
Let $\alpha>0$ and $N\in\N$. 
Define the Hankel integral operator $G_{N,\alpha}: L^2(\R^+)\to L^2(\R^+)$,
\begin{equation*}
  (G_{N,\alpha}\varphi)(x) = \int_0^\infty G_{N,\alpha}(x+y)\varphi(y)\, dy,\ x\in\R^+,
\end{equation*}
with kernel function
\begin{equation*}
  G_{N,\alpha}(x) \coloneqq e^{-\frac{\alpha}{2}x} \sum_{j=0}^N e^{-jx} 
    = e^{-\frac{\alpha}{2}x} \frac{e^{\frac{x}{2}}}{2\sinh(\frac{x}{2})} (1-e^{-Nx}) .
\end{equation*}
Then, $\sigma(H_{N,\alpha})\setminus\{0\} = \sigma(G_{N,\alpha})\setminus\{0\}$. In particular, $\|G_{N,\alpha}\|=\|H_{N,\alpha}\|$.
\end{lemma}
\begin{proof}
With the functions
\begin{equation*}
  e_j \in L^2(\R^+),\ e_j(x) = e^{-jx-\frac{\alpha}{2}x},\ j\in\N_0,
\end{equation*}
we define the operators $A:L^2(\R^+)\to \C^N$ and $B:\C^N\to L^2(\R^+)$,
\begin{equation*}
  (A\varphi)_j = \int_0^\infty e_j(x)\varphi(x)\, dx,\ j=0,\ldots, N-1,
  (Bc)(x) = \sum_{j=0}^{N-1} c_j e_j(x),\ x\in\R^+ .
\end{equation*}
It is easily checked that $H_{N,\alpha} = AB:\C^N\to\C^N$. On the other hand, $BA:L^2(\R^+)\to L^2(\R^+)$,
\begin{equation*}
  (BA\varphi)(x) = \sum_{j=0}^{N-1} e_j(x) \int_0^\infty e_j(y)\varphi(y)\, dy 
    = \int_0^\infty \varphi(y)\sum_{j=0}^{N-1} e_j(x) e_j(y)\, dy
    = (G_{N,\alpha}\varphi)(x)
\end{equation*}
since $e_j(x)e_j(y) = e_j(x+y)$. Now, $\sigma(AB)\setminus\{0\} = \sigma(BA)\setminus\{0\}$ which completes the proof.
\end{proof}

We extract the asymptotically relevant part of the operator $G_{N,\alpha}$.
This gives rise to orthogonal projections generated by characteristic functions. Throughout, we will use the notation
\begin{equation}\label{characteristic_function}
  P_{[a,b]}: L^2\to L^2,\ (P_{[a,b]}\varphi)(x) = \chi_{[a,b]}(x)\varphi(x)
\end{equation}
where $\chi_{[a,b]}$ is the characteristic function of the interval $[a,b]$.

\begin{lemma}\label{hm02t}
Let $E_\alpha:L^2(\R^+)\to L^2(\R^+)$ be the integral operator with kernel function
\begin{equation}\label{hm02t01}
  E_\alpha(x,y) = e^{-x(y+\frac{\alpha}{2})} .
\end{equation}
Then $E_\alpha$, $\alpha\geq 0$, is bounded with $\|E_\alpha\|\leq \sqrt{\pi}$. Moreover, $E_\alpha P_{[0,N]}E_\alpha$, $\alpha>0$,
is a trace class operator with
\begin{equation}\label{hm02t02}
  \| E_\alpha P_{[0,N]} E_\alpha^*\|_1 = \frac{1}{2} \ln\bigl( \frac{2N+\alpha}{\alpha} \bigr) .
\end{equation}
The difference
\begin{equation*}
  D_N \coloneqq G_{N,\alpha} - E_\alpha P_{[0,N]} E_\alpha^*
\end{equation*}
is in the trace class with $\|D_N\| \leq C_\alpha< \infty $ for all $N\in\N$.
\end{lemma}
\begin{proof}
We use a generalized version of the Schur test (see \cite[Thm. 5.2]{HalmosSunder1978}) with test functions $p(x)=q(x)=\frac{1}{\sqrt{x}}$.
Then, by standard computations
\begin{equation*}
  \int_0^\infty e^{-x(y+\frac{\alpha}{2})} \frac{1}{\sqrt{y}}\, dy 
     = \frac{\sqrt{\pi}}{\sqrt{x}} e^{-\frac{\alpha x}{2}} \leq \sqrt{\pi}\frac{1}{\sqrt{x}}.
\end{equation*}
Likewise,
\begin{equation*}
  \int_0^\infty \frac{1}{\sqrt{x}} e^{-x(y+\frac{\alpha}{2})} \, dx
     \leq \int_0^\infty \frac{1}{\sqrt{x}} e^{-xy}\, dx = \sqrt{\pi}\frac{1}{\sqrt{y}} .
\end{equation*}
This implies $E_\alpha$ is bounded with the given estimate for the norm.

In order to show the trace class property we start from the simple formula
\begin{equation*}
  1 - e^{-Nx} = x \int_0^N e^{-xt}\, dt
\end{equation*}
and rewrite the kernel function $G_{N,\alpha}$
\begin{equation*}
  G_{N,\alpha}(x) = e^{-\frac{\alpha}{2}x} \frac{e^{\frac{x}{2}x}}{2\sinh(\frac{x}{2})} \int_0^N e^{xt}\, dt
         = \int_0^N e^{-x(t+\frac{\alpha}{2})}\, dt  
             + e^{-\frac{\alpha}{2}x} \bigl[ \frac{e^{\frac{x}{2}}x}{2\sinh(\frac{x}{2})} - 1\bigr] \int_0^N e^{-xt}\, dt .
\end{equation*}
The first term gives rise to the Hankel operator $\tilde G_{N,\alpha}$ with kernel function
\begin{equation*}
  \tilde G_{N,\alpha}(x) = \int_0^N e^{-x(t+\frac{\alpha}{2})}\, dt .
\end{equation*}
We write this as follows (cf. \eqref{hm02t01})
\begin{equation*}
  \tilde G_{N,\alpha}(x+y) = \int_0^N e^{-(x+y)(t+\frac{\alpha}{2})}\, dt
                  = \int_0^N e^{-x(t+\frac{\alpha}{2})} e^{-(t+\frac{\alpha}{2})y}\, dt
                  = \int_0^N E_\alpha(x,t) E_\alpha(y,t)\, dt
\end{equation*}
which implies $\tilde G_{N,\alpha} = E_\alpha P_{[0,N]} E_\alpha^*$. Since, obviously, $E_\alpha P_{[0,N]}E_\alpha^*\geq 0$ we obtain
\begin{equation*}
  \| E_\alpha P_{[0,N]} E_\alpha^* \|_1
     = \tr(E_\alpha P_{[0,N]} E_\alpha^* )
     = \int_0^\infty \int_0^N e^{-2x(y+\frac{\alpha}{2})}\, dy\, dx
     = \int_0^N \frac{1}{\alpha+2y}\, dy
     = \frac{1}{2}\ln\bigl( \frac{\alpha + 2N}{\alpha} \bigr).
\end{equation*}
The remaining difference is the Hankel operator $D_N$ with kernel function
\begin{equation*}
  D_N(x) \coloneqq e^{-\frac{\alpha}{2}x} \bigl[ \frac{e^{\frac{x}{2}}x}{2\sinh(\frac{x}{2})} - 1\bigr] \int_0^N e^{-xt}\, dt 
          =  \bigl[ \frac{e^{\frac{x}{2}}x}{2\sinh(\frac{x}{2})} - 1\bigr] \int_{\frac{\alpha}{2}}^{N+\frac{\alpha}{2}} e^{-xt}\, dt .
\end{equation*}
In order to show that $D_N$ is in the trace class we use Howland's criterion \cite[Thm. 2.1]{Howland1971}, which also
gives a bound on the trace norm. To this end, we need the derivative
\begin{equation*}
  D_N'(x) = \biggl\{ \frac{1-e^{-x}-xe^{-x}}{(1-e^{-x})^2} - \bigl[ \frac{x}{1-e^{-x}} - 1\bigr]x\biggr\} 
               \int_{\frac{\alpha}{2}}^{N+\frac{\alpha}{2}} e^{-xt}\, dt .
\end{equation*}
Via the elementary estimates 
\begin{equation*}
  0 \leq \frac{x}{1-e^{-x}} - 1 \leq x,\ 0 \leq \frac{1-e^{-x}-xe^{-x}}{(1-e^{-x})^2}\leq 1\ \text{for}\ x\geq 0 ,
\end{equation*}
we obtain
\begin{equation*}
  |D_N'(x)| \leq (1+x^2) \int_{\frac{\alpha}{2}}^{N+\frac{\alpha}{2}} e^{-xt}\, dt
            \leq (x+\frac{1}{x}) e^{-\frac{\alpha}{2}x} .
\end{equation*}
Then, Howland's criterion shows that $D_N$ is in the trace class with
\begin{equation*}
  \|D_N\|_1 \leq \int_0^\infty x^{\frac{1}{2}} \biggl[ \int_x^\infty |D_N'(y)|^2\, dy \biggr]^{\frac{1}{2}}\, dx
            \leq \int_0^\infty x^{\frac{1}{2}} \biggl[ \int_x^\infty (y^2 + 2 + \frac{1}{y^2})e^{-\alpha y} \, dy\biggr]^{\frac{1}{2}}\, dx 
            \eqqcolon C_\alpha .
\end{equation*}
Elementary estimates show that $C_\alpha<\infty$ for $\alpha>0$.
Note that $C_\alpha$ does not depend on $N$.
\end{proof}

We relate $E_\alpha P_{[0,N]}E_\alpha^*$ to the Carleman operator $K: L^2(\R^+)\to L^2(\R^+)$,
\begin{equation}\label{carleman_operator}
  (K\varphi)(x) = \int_0^\infty \frac{1}{x+y} \varphi(y)\, dy,\ x\in\R^+ .
\end{equation}
It is well-known that $K$ is self-adjoint and satisfies (see \cite[Theorem~8.14]{Peller2003} for the operator norm)
\begin{equation}\label{carleman_operator_norm}
  0\leq K  \leq \pi .
\end{equation}
We define the translation operator
\begin{equation}\label{translation_operator}
  T_\alpha:L^2(\R^+)\to L^2(\R^+),\ 
  (T_\alpha \varphi)(x) = 
\begin{cases}
  \varphi(x-\frac{\alpha}{2}) & \text{for}\ x\geq \frac{\alpha}{2},\\
    0                         & \text{for}\ 0\leq x< \frac{\alpha}{2} .
\end{cases}
\end{equation}
Its pseudo inverse is given by
\begin{equation*}
  T_\alpha^+:L^2(\R^+)\to L^2(\R^+),\ 
  (T_\alpha^+\varphi)(x) = \varphi(x+\frac{\alpha}{2}),\ x\geq 0 .
\end{equation*}
That is to say,
\begin{equation}\label{translation_operator_inverse}
  P_{\interval[open right]{\frac{\alpha}{2}}{\infty}} T_\alpha T_\alpha^+ = P_{\interval[open right]{\frac{\alpha}{2}}{\infty}} .
\end{equation}
We move the $\alpha$ from the integral operator to the projection.

\begin{lemma}\label{hm03t}
Let $\alpha>0$ and $N\in\N$.
The operator $E_\alpha P_{[0,N]}E_\alpha^*$ and the Carleman operator $K$, cf. \eqref{hm02t01} and \eqref{carleman_operator},
satisfy
\begin{equation}\label{hm03t01}
  \sigma(E_\alpha P_{[0,N]} E_\alpha^* )\setminus \{0\} 
    = \sigma( P_{[\frac{\alpha}{2},N+\frac{\alpha}{2}]} K P_{[\frac{\alpha}{2},N+\frac{\alpha}{2}]}) \setminus\{ 0\} .
\end{equation}
\end{lemma}
\begin{proof}
We know that 
\begin{equation*}
  \sigma(E_\alpha P_{[0,N]}E_\alpha^*)\setminus \{0\} =  \sigma(E_\alpha^*E_\alpha P_{[0,N]})\setminus\{0\} .
\end{equation*}
The product $E_\alpha^*E_\alpha$ is a quasi-Carleman operator
\begin{equation*}
  (E_\alpha^*E_\alpha)(x,y) = \int_0^\infty e^{-(x+\frac{\alpha}{2})t} e^{-t(y+\frac{\alpha}{2})}\, dt = \frac{1}{x+y+\alpha} .
\end{equation*}
By using $T_\alpha$ (cf. \eqref{translation_operator})
\begin{equation*}
\begin{split}
  (E_\alpha^*E_\alpha P_{[0,N]}\varphi)(x)
   & = \int_0^N \frac{1}{x+y+\alpha}\varphi(y)\, dy \\
   & = \int_{\frac{\alpha}{2}}^{N+\frac{\alpha}{2}} \frac{1}{x+y+\frac{\alpha}{2}} \varphi(y-\frac{\alpha}{2})\, dy\\
   & = \int_0^\infty \frac{1}{x+y+\frac{\alpha}{2}} _{[\frac{\alpha}{2},N+\frac{\alpha}{2}]} (T_\alpha\varphi)(y)\, dy\\
   & = (T_\alpha^+ K P_{[\frac{\alpha}{2},N+\frac{\alpha}{2}]} T_\alpha \varphi)(x) .
\end{split}
\end{equation*}
In operator form this reads
\begin{equation*}
  E_\alpha^* E_\alpha P_{[0,N]} = T_\alpha^+ K P_{[\frac{\alpha}{2},N+\frac{\alpha}{2}]} T_\alpha 
\end{equation*}
which implies
\begin{equation*}
  \sigma(E_\alpha^*E_\alpha P_{[0,N]}) \setminus \{0\}
   = \sigma( KP_{[\frac{\alpha}{2},N+\frac{\alpha}{2}]} T_\alpha T_\alpha^+) \setminus\{0\}
   = \sigma(KP_{[\frac{\alpha}{2},N+\frac{\alpha}{2}]})\setminus\{0\} .
\end{equation*}
Here we used \eqref{translation_operator_inverse}. This implies \eqref{hm03t01}.
\end{proof}

In order to use the perturbation determinant \eqref{perturbation_determinant} we need a certain inverse.

\begin{lemma}\label{hm04t}
Let $\alpha\geq \frac{1}{2}$.
Furthermore, let $\beta\in\C\setminus\interval[open right]{1}{\infty}$, $s\in [0,1]$, and $N\in\N$. Then, the operator $\id - \beta A_{N,\alpha}(s)$,
\begin{equation*}
  A_{N,\alpha}(s) \coloneqq \frac{1}{\pi}\bigl((1-s)E_\alpha P_{[0,N]} E_\alpha^* + s G_{N,\alpha}\bigr),
\end{equation*}
is invertible with
\begin{equation*}
  \|(\id-\beta A_{N,\alpha}(s)^{-1}\| \leq 
\begin{cases}
  1                            & \text{for}\ \re(\beta)\leq 0, \\
  \frac{1}{1-\re(\beta)}       & \text{for}\ 0< \re(\beta)< 1,\\
  \frac{|\beta|}{|\im(\beta)|} & \text{for}\ \im(\beta)\neq 0.
\end{cases}
\end{equation*}
\end{lemma}
\begin{proof}
We use the Lax--Milgram theorem. Note that Lemmas \ref{hm02t} and \ref{hm01t} along with \eqref{hilbert_matrix_spectrum}
imply $0\leq A_{N,\alpha}(s)\leq \id$ in the sense of quadratic forms. Furthermore,
\begin{equation*}
  \re(\id-\beta A_{N,\alpha}(s)) = \id -\re(\beta)A_{N,\alpha}(s) .
\end{equation*}
Hence, for $\re(\beta)\leq 0$
\begin{equation*}
  \re(\id-\beta A_{N,\alpha}(s)) \geq \id
\end{equation*}
and for $0<\re(\beta)<1$
\begin{equation*}
  \re(\id-\beta A_{N,\alpha}(s)) \geq (1-\re(\beta))\id\ \text{with}\ 1-\re(\beta)>0,
\end{equation*}
which yield the first two cases. In the third case, surely $\beta\neq 0$. Hence,
\begin{equation*}
  \id - \beta A_{N,\alpha}(s) = \beta( \frac{1}{\beta}\id - A_{N,\alpha}(s))
\end{equation*}
and
\begin{equation*}
  \im(\frac{1}{\beta}\id - A_{N,\alpha}(s)) = - \frac{\beta}{|\beta|^2}\id.
\end{equation*}
This implies that the inverse exists and is bounded with
\begin{equation*}
  \| (\id - \beta A_{N,\alpha}(s))^{-1}\|
     = \frac{1}{|\beta|} \|(\frac{1}{\beta}\id - A_{N,\alpha}(s))^{-1}\|
     \leq \frac{1}{|\beta|} \frac{|\beta|^2}{|\im(\beta)|} .
\end{equation*}
This completes the proof.
\end{proof}

The asymptotics of the determinant under study is given by the corresponding determinant of the Carleman operator.

\begin{proposition}\label{hm05t}
Let $\alpha>0$ and $N\in\N$. 
The operator $P_{[\frac{\alpha}{2},N+\frac{\alpha}{2}]}KP_{[\frac{\alpha}{2},N+\frac{\alpha}{2}]}: L^2(\R^+)\to L^2(\R^+)$, cf. \eqref{carleman_operator}, is in the trace class.
Furthermore, if $\alpha\geq\frac{1}{2}$ and $\beta\in\C\setminus\interval[open right]{1}{\infty}$,
\begin{equation*}
  \det(\id - \frac{\beta}{\pi} H_{N,\alpha})
     = \det(\id - \frac{\beta}{\pi} P_{[\frac{\alpha}{2},N+\frac{\alpha}{2}]}KP_{[\frac{\alpha}{2},N+\frac{\alpha}{2}]})
        \times \Delta_N(\beta)
\end{equation*}
where the perturbation determinant can be bounded as
\begin{equation*}
  \exp[ - C(\beta) |\beta| \|D_N\|_1 ] \leq |\Delta_N(\beta)| \leq \exp[ C(\beta)|\beta| \|D_N\|_1 ]
\end{equation*}
with $0\leq C(\beta)<\infty$ independent of $N$, cf. Lemmas \ref{hm04t} and \ref{hm02t}.
\end{proposition}
\begin{proof}
The trace class property follows immediately from Lemmas \ref{hm03t} and \ref{hm02t}. We apply the
formula \eqref{perturbation_determinant} for the perturbation determinant to the operator (cf. Lemma \ref{hm02t})
\begin{equation*}
  G_{N,\alpha} = E_\alpha P_{[0,N]} E_\alpha^* + D_N
\end{equation*}
thereby obtaining
\begin{equation*}
  \det(\id - \frac{\beta}{\pi} G_{N,\alpha})
    = \det(\id - \frac{\beta}{\pi} E_\alpha P_{[0,N]}E_\alpha^*) \times \Delta_N(\beta) .
\end{equation*}
Using the formula \eqref{determinant_integral} for the perturbation determinant we write this as
\begin{equation*}
  \Delta_N(\beta) 
    = \exp\biggl[ - \frac{\beta}{\pi} \int_0^1 \tr\biggl\{ \bigl[ \id - (1-s)\frac{\beta}{\pi} E_\alpha P_{[0,N]} E_\alpha^*
                      -s\frac{\beta}{\pi} G_{N,\alpha}\bigr]^{-1} D_N \biggr\} \, ds\biggr] .
\end{equation*}
Finally, we bound the trace by the trace norm and use Lemma \ref{hm04t} to estimate the norm of the inverse.
This completes the proof.
\end{proof}

\section{\texorpdfstring{Szeg\H{o} limit theorem}{Szego limit theorem}\label{slt}}
In order to handle the complex parameter $\beta$ we formulate an abstract Szeg\H{o} theorem
for normal operators based upon \cite{Otte2004} and \cite{BoettcherOtte2005}.

\begin{proposition}\label{slt01t}
Let $A:\hilbert\to\hilbert$ be a bounded normal operator with
\begin{equation}\label{slt01t01}
  \re(\lambda)\geq m,\ \im(\lambda)\in[y_0-h,y_0+h]\ \text{for all}\ \lambda\in\sigma(A)
\end{equation}
where $m\in\R$ and $0\leq h<\frac{\pi}{2}$.
Furthermore, let $P:\hilbert\to\hilbert$ be an orthonormal projection such that $PAP$ is in the trace class. Then,
the determinant of the operator $Pe^AP:\ran(P)\to\ran(P)$ satisfies
\begin{equation}\label{slt01t02}
  \det(Pe^AP) = \exp[ \tr(PAP)+ \rho(A) ]
\end{equation}
where
\begin{equation}\label{slt01t03}
  |\rho(A)| \leq \frac{1}{2} \frac{e^{|m|}}{\cos(h)} e^{\|A\|} \|PA(\id-P)\|_2 \|(\id - P)AP\|_2 .
\end{equation}
\end{proposition}
\begin{proof}
From (19) in \cite{Otte2004} follows
\begin{equation*}
  |\rho(A)| \leq e^{\|A\|} \|PA(\id-P)\|_2 \|(\id-P)AP\|_2 \int_0^1 t \| (Pe^{tA}P)^{-1}P\|\, dt .
\end{equation*}
From (15) and (16) in \cite{BoettcherOtte2005} we infer
\begin{equation*}
  \| (Pe^{tA}P)^{-1}P \| \leq \frac{e^{|m|}}{\cos(h)},\ 0\leq t\leq 1 ,
\end{equation*}
which proves the statement.
\end{proof}

In the special case when 
$A:L^2(\R)\to L^2(\R)$ is a convolution operator with even kernel function $A(x)=A(-x)$
and $P=P_{[-n,n]}$ 
the Hilbert--Schmidt norm appearing in Proposition \ref{slt01t} can be written after some simple calculations
\begin{equation}\label{kernel_estimate}
\begin{split}
  \|P_{[-n,n]}A(\id-P_{[-n,n]})\|_2^2
    & = \int_{|x|\leq n}\int_{|y|\geq n} |A(x-y)|^2\, dy\, dx\\
    & = 2\int_0^n x|A(x)|^2 \, dx + 2n \int_n^\infty |A(x)|^2\, dx + 2 \int_0^n \int_n^\infty |A(x+y)|^2\, dy\, dx .
\end{split}
\end{equation}
In order to apply the abstract result in Proposition \ref{slt01t} to our case, we have to write the operator at hand as
$\id - \frac{\beta}{\pi}K = e^A$. In other words we need a logarithm which is no problem here since the Carleman operator $K$ can
be diagonalized explicitly by means of the Mellin transform. For our purposes it is more convenient to stop halfway
and transform it into a convolution operator.

\begin{lemma}\label{slt02t}
The operator $W_a: L^2(\R^+)\to L^2(\R)$, $a\in\R$,
\begin{equation}\label{slt02t01}
  (W_a\varphi)(s) = \sqrt{2} e^{s+a} \varphi(e^{2s+2a}),\ s\in\R,\ \varphi\in L^2(\R^+)
\end{equation}
is unitary. It transforms the Carleman operator $K$ into a convolution operator
\begin{equation}\label{slt02t02}
  W_a K W_a^* = K_0,\ K_0: L^2(\R)\to L^2(\R),\ K_0(x-y) = \frac{1}{\cosh(x-y)}
\end{equation}
and the projection with $a=\frac{1}{4}(\ln(N+\frac{\alpha}{2})+ \ln(\frac{\alpha}{2}))$
\begin{equation}\label{slt02t03}
  W_a P_{[\frac{\alpha}{2}, N+\frac{\alpha}{2}]} W_a^* = P_{[-n_{\frac{\alpha}{2}}(N),n_{\frac{\alpha}{2}}(N)]},\ n_{\frac{\alpha}{2}}(N) = \frac{1}{4} \ln(\frac{N+\frac{\alpha}{2}}{\frac{\alpha}{2}}).
\end{equation}
\end{lemma}
\begin{proof}
Cf. \cite[Ch. 10, Thm. 2.1]{Peller2003} and also \cite{Yafaev2013}. We will use the substitution
\begin{equation*}
  x = e^{2s+2a},\ dx = 2e^{2s+2a}\, ds .
\end{equation*}
The unitarity follows from, $\varphi\in L^2(\R^+)$,
\begin{equation*}
  \|W_a\varphi\|^2
     = 2\int_\R |\varphi(e^{2s+2a})|^2 e^{2s+2a}\, ds = \int_0^\infty |\varphi(x)|^2\, dx = \|\varphi\|^2 
\end{equation*}
and the analogous calculation for $W_a^*$. For the Carleman operator we obtain
\begin{equation*}
\begin{split}
  (W_aK\varphi)(s)
    & = \sqrt{2} e^{s+a} \int_0^\infty \frac{1}{e^{2s+2a}+y} \varphi(y)\, dy \\
    & = \sqrt{2} e^{s+a} \int_\R \frac{2 e^{2t+2a}}{e^{2s+2a} + e^{2t+2a}} \varphi(e^{2t+2a})\, dt\\
    & = \sqrt{2} \int_\R \frac{2}{e^{s-t}+e^{t-s}} e^{t+a} \varphi(e^{2t+2a})\, dt\\
    & = \int_\R \frac{1}{\cosh(s-t)} (W_a\varphi)(t)\, dt\\
    & = (K_0W_a\varphi)(s)
\end{split}
\end{equation*}
which reads in operator form
\begin{equation*}
  W_a K = K_0 W_a .
\end{equation*}
This yields \eqref{slt02t02}. Finally,
\begin{equation*}
  \chi_{[\frac{\alpha}{2},\N+\frac{\alpha}{2}]}(e^{2s+2a}) =
\begin{cases}
  1 & \text{for}\ \frac{\alpha}{2} \leq e^{2s+2a} \leq N+\frac{\alpha}{2}, \\
  0 & \text{otherwise},
\end{cases}
  =
\begin{cases}
  1 & \text{for}\ \frac{1}{2}\ln(\frac{\alpha}{2}) - a \leq s \leq \frac{1}{2}\ln(N+\frac{\alpha}{2}) - a ,\\
  0 & \text{otherwise} .
\end{cases}
\end{equation*}
The special $a$ yields the formula \eqref{slt02t03} for the projection.
\end{proof}

Via the Fourier transform
\begin{equation}\label{fourier_transform}
  (\mathcal{F}\varphi)(\omega) \coloneqq \hat\varphi(\omega) \coloneqq \frac{1}{\sqrt{2\pi}} \int_\R e^{-i\omega x}\varphi(x)\, dx
\end{equation}
the convolution operator $K_0$ can be transformed into a multiplication operator
\begin{equation}\label{symbol}
  \mathcal{F}K_0\varphi = \sqrt{2\pi}\hat K_0 \hat\varphi,\ \hat K_0(\omega) = \sqrt{\frac{\pi}{2}} \frac{1}{\cosh(\frac{\omega\pi}{2})} .
\end{equation}
Thereby, we can construct the logarithm needed for the Szeg\H{o} theorem.

\begin{lemma}\label{slt03t}
Let $\beta\in\C\setminus\interval[open right]{1}{\infty}$ and
let the convolution operator $A_0: L^2(\R)\to L^2(\R)$ be given by its kernel function
\begin{equation}\label{slt03t01}
  A_0(x) = \frac{1}{2\pi} \int_\R e^{i\omega x} \ln\bigl( 1 - \frac{\beta}{\cosh(\frac{\omega\pi}{2})}\bigr)\, d\omega,\
  \hat A_0(\omega) = \frac{1}{\sqrt{2\pi}} \ln\bigl( 1 - \frac{\beta}{\cosh(\frac{\omega\pi}{2})}\bigr) .
\end{equation}
Then,
\begin{equation}\label{slt03t02}
  e^{A_0} = \id - \frac{\beta}{\pi} K_0 .
\end{equation}
Furthermore, the spectrum $\sigma(A_0)$ of $A_0$ satisfies
\begin{equation}\label{slt03t03}
\begin{gathered}
  \{ \re(\lambda) \mid \lambda\in\sigma(A_0)\} = [m,M]\ \text{with}\ M=\max\{ 0,\ln|1-\beta|\},\\
  m=
\begin{cases}
  \ln|1-\frac{\re(\beta)}{\beta}| & \text{if}\ 0\leq \re(\beta)\leq |\beta|^2,\\
   \min\{0,\ln|1-\beta|\}         & \text{otherwise},
\end{cases}
\end{gathered}
\end{equation}
and
\begin{equation}
\begin{gathered}\label{slt03t04}
  \{ \im(\lambda)\mid \lambda\in\sigma(A_0)\} = [y_0-h,y_0+h],
      \ y_0=\frac{1}{2}a(\beta),\ h=\frac{1}{2}|a(\beta)| < \frac{\pi}{2} ,\\
   a(\beta) = -\sign(\im(\beta)) \biggl[ \frac{\pi}{2} - \arctan(\frac{1-\re(\beta)}{|\im(\beta)|}) \biggr] .
\end{gathered}
\end{equation}
\end{lemma}
\begin{proof}
To get all the $\pi$'s right note that \eqref{slt03t02} is, via the Fourier transform (cf. \eqref{fourier_transform}), equivalent to
\begin{equation*}
  \exp(\sqrt{2\pi} \hat A_0(\omega)) = 1 - \frac{\beta}{\pi}\sqrt{2\pi}\hat K_0(\omega) .
\end{equation*}
Solving for $\hat A_0(\omega)$ and using \eqref{symbol} for $\hat K_0(\omega)$
as well as the inverse Fourier transform prove \eqref{slt03t01}.

The spectrum of $A_0$ is given up to factor through the numerical range of the function $\hat A_0$
\begin{equation*}
  \sigma(A_0) = \{ \ln( 1 - \frac{\beta}{\cosh(\frac{\omega\pi}{2})} ) \mid \omega\in\R \} \cup \{0\}
              = \{ \ln(1-s\beta)\mid 0\leq s\leq 1\}.
\end{equation*}
Using the the principal branch of the logarithm as in \eqref{logarithm} yields
\begin{equation*}
  \ln(1-s\beta) = -\beta\int_0^s \frac{1}{1-\beta t} \, dt
                = -\int_0^s \frac{\beta-|\beta|^2t}{|1-\beta t|^2}\, dt .
\end{equation*}
The imaginary part is
\begin{equation*}
  \im(\ln(1-s\beta)) = -\im(\beta) \int_0^s \frac{1}{|1-\beta t|^2}\, dt .
\end{equation*}
The integral vanishes at $s=0$ and attains its maximal value at $s=1$. For $\im(\beta)\neq 0$ we obtain after some standard 
substitutions
\begin{equation*}
  \im(\ln(1-\beta)) = -\im(\beta) \int_1^\infty \frac{1}{|t-\beta|^2}\, dt
                    = -\sign(\im(\beta)) \int_{\frac{1-\re(\beta)}{|\im(\beta)|}}^\infty \frac{1}{t^2+1}\, dt
\end{equation*}
and for the remaining case
\begin{equation*}
  \im(\ln(1-\frac{\beta}{s})) = 0 \ \text{for}\ \im(\beta)=0 .
\end{equation*}
this implies \eqref{slt03t04}. The bound $h\leq \frac{\pi}{2}$ is obvious.
Since $h=\frac{\pi}{2}$ would require $1-\re(\beta)<0$ and $\im(\beta)=0$ this cannot
occur due to the assumptions on $\beta$.

The real part is
\begin{equation*}
  \re(\ln(1-s\beta)) = - \int_0^s \frac{\re(\beta)-|\beta|^2t}{|1-\beta t|^2}\, dt
                     = \ln|1-s\beta|
                     \eqqcolon f(s) .
\end{equation*}
For those $\beta$'s satisfying
\begin{equation*}
  0\leq \re(\beta) \leq |\beta|^2
\end{equation*}
the function $f$ has a single local extremum at $s_-\in[0,1]$, which is a minimum with
\begin{equation*}
  f(s_-) = \ln\bigl| 1 - \frac{\re(\beta)}{\beta}\bigr| = \ln\bigl( \frac{|\im(\beta)|}{|\beta|}\bigr) \leq 0.
\end{equation*}
For any other $\beta$ the extremal values are given by $f(0)=0$ and $f(1) = \ln|1-\beta|$. This proves \eqref{slt03t03}.
\end{proof}

We apply Proposition \ref{slt01t} to the operator $K_0$.

\begin{proposition}\label{slt04t}
Let $\beta\in\C\setminus\interval[open right]{1}{\infty}$ and $n\geq 0$. Then for $K_0$ from \eqref{slt02t02},
\begin{equation}\label{slt04t01}
  \det\bigl(\id - \frac{\beta}{\pi} P_{[-n,n]} K_0 P_{[-n,n]}\bigr) = \exp[ 2n\gamma(\beta) + \rho_n ] .
\end{equation}
Here
\begin{equation}\label{slt04t02}
  \gamma(\beta) %= \frac{1}{2\pi} \int_\R \ln\bigl( 1 - \frac{\beta}{\cosh(\frac{\omega\pi}{2})})\, d\omega
                = \frac{1}{\pi^2} \int_0^\infty \ln\bigl( 1-\frac{\beta}{\cosh(\omega)}\bigr)\, d\omega
                = \frac{1}{\pi^2} \bigl[ \arcosh(-\beta) \bigr]^2 + \frac{1}{4}
\end{equation}
and the correction term satisfies (cf. \eqref{slt03t01})
\begin{equation*}
  |\rho_n | \leq \frac{3}{4\pi} \frac{e^{|m|}}{\cos(h)} \bigl( \|\hat A_0\|_1 + \|\hat A_0''\|_1\bigr)^2
\end{equation*}
with $m$ from \eqref{slt03t03} and $0\leq h<\frac{\pi}{2}$ from \eqref{slt03t04}.
\end{proposition}
\begin{proof}
We check the conditions of Proposition \ref{slt01t}.
The second part of \eqref{slt01t01} follows immediately from \eqref{slt03t04}
since $0\leq h<\frac{\pi}{2}$ for $\beta\in\C\setminus\interval[open right]{1}{\infty}$.
For the real part the only critical cases in \eqref{slt03t03}
are $\beta=1$ and $\frac{\re(\beta)}{\beta}=1$, which is equivalent to $\beta=1$. Since $\beta\notin\interval[open left]{1}{\infty}$
this cannot occur. Hence, there is an $m\in\R$ with $|m|<\infty$ such that the first part in \eqref{slt01t01} holds true.

In order to bound the correction $\rho_n$ term we use \eqref{kernel_estimate}. Since $A_0$ is the Fourier transform of an $L^1$-function
$\hat A_0$ that is arbitrarily often differentiable and vanishes at infinity appropriately, cf. \eqref{slt03t01},
a simple integration by parts shows
\begin{equation*}
  |A_0(x)| \leq \frac{1}{\sqrt{2\pi}} \frac{1}{1+x^2} \bigl[ \|\hat A_0\|_1 + \|\hat A_0''\|_1 \bigr],\ x\in\R.
\end{equation*}
For $\beta\notin \interval[open left]{1}{\infty}$ the $L^1$-norms are finite which follows most conveniently from the representation
\begin{equation*}
  \hat A_0(\omega) = - \frac{\beta}{\sqrt{2\pi}} \int_0^1 \frac{1}{\cosh(\frac{\omega\pi}{2}) - t\beta}\, dt
\end{equation*}
and the analogous formula for $\hat A_0''(\omega)$. The integrals in \eqref{kernel_estimate} become in our case
\begin{gather*}
  2 \int_0^n \frac{x}{(1+x^2)^2}\, dx \leq 1,\
  2n\int_n^\infty \frac{1}{(1+x^2)^2}\, dx \leq \int_n^\infty \frac{2x}{(1+x^2)^2}\, dx = \frac{1}{1+n^2},\\
  2 \int_0^n \int_n^\infty \frac{1}{(1+(x+y)^2)^2} \, dy\, dx \leq 2n\int_n^\infty \frac{1}{(1+y^2)^2}\, dy \leq \frac{1}{1+n^2}.
\end{gather*}
Thereby,
\begin{equation*}
  \|P_{[-n,n]}A_0(\id-P_{[-n,n]})\| \cdot \|(\id-P_{[-n,n]})A_0P_{[-n,n]}\| \leq\frac{3}{2\pi}[ \|\hat A_0\|_1 + \|\hat A_0''\|_1]
\end{equation*}
Finally, the leading term in \eqref{slt01t02} is
\begin{equation*}
  \tr( P_{[-n,n]} A_0 P_{[-n,n]} )  = 2n A_0(0)
\end{equation*}
Now,
\begin{equation*}
  A_0(0) = \frac{1}{2\pi} \int_\R \ln\bigl( 1 - \frac{\beta}{\cosh(\frac{\omega\pi}{2})}\bigr) \, d\omega
         = \frac{2}{\pi^2} \int_0^\infty \ln\bigl( 1 - \frac{\beta}{\cosh(\omega)}\bigr) \, d\omega
         = \frac{1}{\pi^2} [\arcosh (-\beta) ]^2 + \frac{1}{4}
\end{equation*}
where we evaluated the integral via Lemma \ref{int02t}.
This completes the proof.
\end{proof}

We summarize our findings by formulating the main result, the Szeg\H{o} limit theorem for the Hilbert matrix.

\begin{theorem}\label{slt05t}
Let $\alpha\geq\frac{1}{2}$ and $N\in\N$. Then, the Hilbert matrix $H_{N,\alpha}$, see \eqref{hilbert_matrix}, satisfies for all
$\beta\in\C\setminus\interval[open right]{1}{\infty}$
\begin{equation}\label{slt05t01}
  \det\bigl(\id - \frac{\beta}{\pi} H_{N,\alpha}\bigr) = \exp[ 2n_{\frac{\alpha}{2}}(N) \gamma(\beta) + O(1)]\ \text{as}\ N\to\infty,\
   n_{\frac{\alpha}{2}}(N) = \frac{1}{4} \ln\bigl(\frac{N+\frac{\alpha}{2}}{\frac{\alpha}{2}}\bigr),
\end{equation}
with the coefficient
\begin{equation}\label{slt05t02}
  \gamma(\beta) = \frac{1}{\pi^2} [\arcosh(-\beta)]^2 + \frac{1}{4} .
\end{equation}
\end{theorem}
\begin{proof}
From Proposition \ref{hm05t} we know
\begin{equation*}
  \ln\bigl( \det( \id - \frac{\beta}{\pi}H_{N,\alpha}) \bigr) 
    = \ln\bigl( \det(\id - \frac{\beta}{\pi} P_{[-\frac{\alpha}{2},N+\frac{\alpha}{2}]} K P_{[-\frac{\alpha}{2},N+\frac{\alpha}{2}]} )\bigr) + O(1)
\end{equation*}
with the Carleman operator $K$ from \eqref{carleman_operator}. From Proposition \ref{slt04t} we infer
\begin{equation*}
  \det(\id - \frac{\beta}{\pi}  P_{[-\frac{\alpha}{2},N+\frac{\alpha}{2}]} K  P_{[-\frac{\alpha}{2},N+\frac{\alpha}{2}]} ) 
       = \det(\id - P_{[-n_{\frac{\alpha}{2}}(N),n_{\frac{\alpha}{2}}(N)]} K_0 P_{[-n_{\frac{\alpha}{2}}(N),n_{\frac{\alpha}{2}}(N)]}).
\end{equation*}
The Szeg\H{o} theorem for $K_0$, Proposition \ref{slt04t}, is
\begin{equation*}
  \ln\bigl( \det(\id - P_{[-n_{\frac{\alpha}{2}}(N),n_{\frac{\alpha}{2}}(N)]} K_0 P_{[-n_{\frac{\alpha}{2}}(N),n_{\frac{\alpha}{2}}(N)]}) \bigr) = 2n_{\frac{\alpha}{2}}(N)\gamma(\beta) + O(1) .
\end{equation*}
Combining these formulae proves the theorem.
\end{proof}

Though the result in \cite{FedeleGebert2019} looks a bit different from ours it is actually the same. For,
\begin{equation*}
  \arcosh(x) = i \arccos(x),\ \arccos(x) = \frac{\pi}{2} - \arcsin(x),\ x\in[-1,1] .
\end{equation*}
These imply
\begin{equation*}
  \frac{1}{\pi^2} ( \arcosh(-\beta))^2 + \frac{1}{4}
    = -\frac{1}{\pi^2}( \frac{\pi}{2} - \arcsin(-\beta))^2 + \frac{1}{4}
    = -\frac{1}{\pi^2} ( \arcsin(\beta)^2 + \pi \arcsin(\beta))
\end{equation*}
which yields the asymptotic formula from \cite[(1.5)]{FedeleGebert2019}
\begin{equation}\label{fedele_gebert}
  \ln\bigl(\det(\id - \frac{\beta}{\pi} H_{N,\alpha})\bigr)
   \sim - \frac{1}{2\pi^2} \bigl( [\arcsin(\beta)]^2 + \pi\arcsin(\beta)\bigr) \ln(N)\
   \text{as}\ N\to\infty .
\end{equation}
By a simple argument based upon the roots of unity we extend our Szeg\H{o} theorem to even powers of the Hilbert matrix.
This will be used for the limit case $\beta=1$, which is not covered by Theorem \ref{slt05t}.

\begin{corollary}\label{slt06t}
Let $m\in\N$ and $\alpha\geq\frac{1}{2}$. Then, the Hilbert matrix $H_{N,\alpha}$ satisfies
\begin{equation}\label{slt06t01}
  \det\bigl(\id + \frac{1}{\pi^{2m}} H_{N,\alpha}^{2m}\bigr) = \exp[ 2n_{\frac{\alpha}{2}}(N) \gamma_{2m} + O(1)] \ \text{as}\ N\to\infty,\
      n_{\frac{\alpha}{2}}(N)=\frac{1}{4}\ln\bigl( \frac{N+\frac{\alpha}{2}}{\frac{\alpha}{2}} \bigr) ,
\end{equation}
where
\begin{equation}\label{slt06t02}
  \gamma_{2m} = \frac{2}{\pi^2}\int_0^\infty \ln\bigl( 1 + \frac{1}{\cosh(\omega)^{2m}}\bigr)\, d\omega .
\end{equation}
\end{corollary}
\begin{proof}
Let us define
\begin{equation*}
  \eta_k = \frac{2k-1}{2m},\ k=1,\ldots, m ,
\end{equation*}
whereby we can factorize the determinant into
\begin{equation*}
  \det( \id + \frac{1}{\pi^{2m}}H_{N,\alpha}^{2m}) =
    \prod_{k=1}^m \det(\id + \frac{1}{\pi}e^{i\pi\eta_k} H_{N,\alpha})
    \prod_{k=1}^m \det(\id + \frac{1}{\pi}e^{-i\pi\eta_k} H_{N,\alpha}) .
\end{equation*}
Note that $e^{\pm i\eta_k}\neq -1$. Therefore, we may apply Theorem \ref{slt05t} to each factor 
in the product which yields for the leading term in the asymptotics
\begin{equation*}
  \gamma_{2m} = \frac{2}{\pi^2} \biggl\{ \sum_{k=1}^m \int_0^\infty \ln\bigl( 1 - \frac{e^{i\pi\eta_k}}{\cosh(\omega)}\bigr)\, d\omega
                 + \sum_{k=1}^m \int_0^\infty \ln\bigl( 1 - \frac{e^{-i\pi\eta_k}}{\cosh(\omega)}\bigr)\, d\omega \biggr\} .
\end{equation*}
Here we used the integral representation \eqref{slt04t02} for the coefficients.
In order to rewrite this we note that
\begin{equation*}
  \ln(z) + \ln(\bar z) = \ln(z\bar z) = \ln(|z|)\ \text{for all}\ z\in\C
\end{equation*}
which implies
\begin{equation*}
  \gamma_{2m} = \frac{2}{\pi^2} \int_0^\infty \ln\Bigl[ \prod_{k=-m}^m ( 1 - \frac{e^{i\pi\eta_k}}{\cosh(\omega)} ) \Bigr]\, d\omega
\end{equation*}
and thus \eqref{slt06t02}.
Since the product is finite the sum of the $O(1)$ terms in \eqref{slt05t01} is still $O(1)$ which shows \eqref{slt06t01}.
\end{proof}

\section{Limit case\label{lc}}
We treat the limit case $\beta=1$, which was not covered by Theorem \ref{slt05t},
by showing that it is the limit, hence the name, of the asymptotics for admissible $\beta$.
More precisely, we provide an upper and lower bound for the asymptotics. The upper bound is straightforward.

\begin{proposition}\label{lc00t}
Let $\alpha\geq\frac{1}{2}$ and $N\in\N$. Then,
\begin{equation}\label{lc00t01}
  \limsup_{N\to\infty} \frac{1}{2n_{\frac{\alpha}{2}}(N)} \ln\det(\id - \frac{1}{\pi}H_{N,\alpha})\leq \gamma(1),\
    n_{\frac{\alpha}{2}}(N)=\frac{1}{4}\ln\Bigl( \frac{N+\frac{\alpha}{2}}{\frac{\alpha}{2}} \Bigr) .
\end{equation}
\end{proposition}
\begin{proof}
Let $\beta<1$. Since $H_{N,\alpha}\geq 0$,
\begin{equation}
  \det(\id - \frac{1}{\pi}H_{N,\alpha}) \leq \det(\id - \frac{\beta}{\pi} H_{N,\alpha}) .
\end{equation}
We already know the asymptotics for these $\beta$'s from Theorem \ref{slt05t}
\begin{equation*}
  \limsup_{N\to\infty} \frac{1}{2n_{\frac{\alpha}{2}}(N)} \ln\bigl( \det(\id - \frac{1}{\pi}H_{N,\alpha}) \bigr)
     \leq \liminf_{N\to\infty} \frac{1}{2n_{\frac{\alpha}{2}}(N)} \ln\bigl( \det(\id - \frac{\beta}{\pi} H_{N,\alpha}) \bigr)
       = \gamma(\beta).
\end{equation*}
Since this is valid for all $\beta<1$ and, moreover, $\gamma(\beta)\to\gamma(1)$ as $\beta\to 1$ we obtain
\eqref{lc00t01}
\end{proof}

For the lower bound we employ Lemma \ref{det01t}. To this end, we need estimates for $\tr(H_{N,\alpha}^m)$.
The method parallels that of Section \ref{hm} in that we replace the Hilbert matrix by the Carleman operator.
For an intermediate step we need the so-called 'odd' Hilbert matrix
\begin{equation}\label{odd_hilbert_matrix}
  H_-: \ell^2(\N_0) \to \ell^2(\N_0),\ H_- = ( h_{j+k} )_{j,k\in\N_0},\
  h_j =
\begin{cases}
  \frac{1}{j+1} & \text{for}\ j \ \text{even}, \\
  0             & \text{for}\ j \ \text{odd} .
\end{cases}
\end{equation}
It is more convenient here to work with the projection operator
\begin{equation}\label{projection}
  P_N : \ell^2(\N_0)\to\ell^2(\N_0),\
  (P_Nc)_j =
\begin{cases}
  c_j & \text{for} \ 0\leq j\leq N-1, \\
  0   & \text{for} \ j\geq N+1
\end{cases}
\end{equation}  
instead of the finite odd Hilbert matrix.

\begin{lemma}\label{lc01t}
Let $\alpha\geq \frac{1}{2}$. Then, for all $m,N\in\N$
\begin{equation}\label{lc01t01}
  \tr[ H_{N,\alpha}^m ] \leq 2^m \tr[ (P_{2N}H_- )^m ] \leq 2^m \tr[ P_{2N} H_-^m ] .
\end{equation}
\end{lemma}
\begin{proof}
We start with the odd Hilbert matrix
\begin{equation*}
\begin{split}
  \tr[ (P_{2N} H_- )^m ]
    & = \sum_{j_1,\ldots,j_m=0}^{2N-1} \prod_{l=1}^m h_{j_l + j_{l+1}}\\
    & = \sum_{k_1,\ldots,k_m=0}^{N-1} \prod_{l=1}^m \frac{1}{2k_l + 2k_{l+1} + 1}
        + \sum_{k_1,\ldots,k_m=0}^{N-1} \prod_{l=1}^m \frac{1}{2k_l +1 + 2k_{l+1} + 1 + 1}\\
    & \geq \frac{1}{2^m} \sum_{k_1,\ldots,k_m=0}^{N-1} \prod_{l=1}^m \frac{1}{k_l + k_{l+1} + \frac{1}{2}}\\
    & \geq \frac{1}{2^m} \sum_{k_1,\ldots,k_m=0}^{N-1} \prod_{l=1}^m \frac{1}{k_l + k_{l+1} + \alpha}\\
    & = \frac{1}{2^m} \tr[ H_{N,\alpha}^m ] .
\end{split}  
\end{equation*}
Here we used that $h_{j_l+j_{l+1}}\neq 0$ only if $j_l+j_{l+1}$ is even which is the case when either all of the $j_l$ are even
or all are odd. This yields the first inequality in \eqref{lc01t01}. The second inequality follows from $P_{2N}$ being an
orthogonal projection and $H_{N,\alpha}^*=H_{N,\alpha}$.
\end{proof}

With the aid of the orthonormal Laguerre functions $l_j$, $j\in\N_0$, we define the unitary operator
\begin{equation}\label{laguerre1}
  U: L^2(\R^+)\to \ell^2(\N_0),\ (U\varphi)_j = \int_0^\infty l_j(x) \varphi(x)\, dx,\ j\in\N_0 .
\end{equation}
This transforms $P_N$ into the projection with Christoffel--Darboux kernel
\begin{equation}\label{laguerre2}
  P_N = U\Pi_NU^*,\ \Pi_N(x,y) \coloneqq \sum_{k=0}^N l_k(x)l_k(y)
\end{equation}
and the odd Hilbert matrix into the Carleman operator \cite[pp.~54,~55]{Peller2003}
\begin{equation}\label{laguerre3}
  2H_- =  UKU^*,\  2^m \tr(P_N H_-^m) = \tr(\Pi_N K^m) .
\end{equation}
The kernel function of the Carleman operator has a critical behavior at $x=0$ and $x=\infty$, cf. \eqref{carleman_operator}.
Therefore, we use an appropriate cut-off.

\begin{lemma}\label{lc02t}
Let $0\leq\delta\leq L$. Then, for all $m,N\in\N$
\begin{equation}\label{lc02t01}
  2^m \tr[ P_N H_-^m ] \leq 2 \tr[ P_{[\delta,L]} K^m] + (1+\pi^m) \tr[ P_{[\delta,L]}^\perp \Pi_N],\
    P_{[\delta,L]}^\perp\coloneqq \id - P_{[\delta,L]} .
\end{equation}
\end{lemma}
\begin{proof}
We use \eqref{laguerre3} and decompose the trace
\begin{equation*}
  \tr ( \Pi_N K^m ) 
     = \tr[ P_{[\delta,L]} \Pi_N P_{[\delta,L]} K^m ] + 2 \re\tr[ P_{[\delta,L]}^\perp \Pi_N P_{[\delta,L]} K^m ] 
          + \tr[P_{[\delta,L]}^\perp \Pi_N P_{[\delta,L]}^\perp K^m ] .
\end{equation*}
Since all operators involved are non-negative we can bound the traces through the operator norm
\begin{equation*}
\begin{split}
  \tr ( \Pi_N K^m ) 
    & \leq \|P_{[\delta,L]} \Pi_N P_{[\delta,L]}\| \tr[P_{[\delta,L]} K^m ]
             + 2 \bigl( \tr(P_{[\delta,L]}^\perp \Pi_N \bigr)^{\frac{1}{2}}
                 \bigl( \tr(P_{[\delta,L]} K^m ) \bigr)^{\frac{1}{2}}
             + \tr(P_{[\delta,L]}^\perp \Pi_N ) \|K\|^m\\
    & \leq \tr[P_{[\delta,L]} K^m ]
             + 2 \bigl( \tr(P_{[\delta,L]}^\perp \Pi_N \bigr)^{\frac{1}{2}}
                 \bigl( \tr(P_{[\delta,L]} K^m ) \bigr)^{\frac{1}{2}}
             + \tr(P_{[\delta,L]}^\perp \Pi_N ) \pi^m  \\
    & \leq 2\tr[P_{[\delta,L]} K^m ] + (1 + \pi^m) \tr(P_{[\delta,L]}^\perp \Pi_N ) .
\end{split}
\end{equation*}
Here we used the Cauchy--Schwarz inequality for the trace and \eqref{carleman_operator_norm}.
This proves the lemma.
\end{proof}

The trace of the Carleman operator can be expressed as a simple integral.

\begin{lemma}\label{lc03t}
Let $\delta>0$ and $N\geq 0$. Then, for all $m\in\N$
\begin{equation*}
  \tr[ P_{[\delta,N+\delta]} K^m ] = 2n_\delta(N)\pi^{m-2} \int_\R \frac{1}{[\cosh(\omega)]^m}\, d\omega,\
    n_\delta(N)=\frac{1}{4}\ln\bigl(\frac{N+\delta}{\delta}\bigr) .
\end{equation*}
\end{lemma}
\begin{proof}
From Lemma \ref{slt02t} we immediately infer
\begin{equation*}
  \tr[ P_{[\delta,N+\delta]} K^m ] = \tr[ P_{[-n_\delta(N),n_\delta(N)]} K_0^m ] .
\end{equation*}
Via the diagonalization $\mathcal{F}K_0\mathcal{F}^* = \sqrt{2\pi} \hat K_0$, see \eqref{fourier_transform}
and \eqref{symbol}, we obtain
\begin{equation*}
  \tr[ P_{[-n_\delta(N),n_\delta(N)]} K_0^m ]
    = (2\pi)^{\frac{m}{2}} \tr[ P_{-n_\delta(N),n_\delta(N)]} \mathcal{F}^* \hat K_0^m \mathcal{F} ]
    = (2\pi)^{\frac{m}{2}} \tr[ \mathcal{F} P_{[-n_\delta(N),n_\delta(N)]} \mathcal{F}^* \hat K_0^m ] .
\end{equation*}
Now,
\begin{equation*}
  \mathcal{F} P_{[-n_\delta(N),n_\delta(N)]} \mathcal{F}^*(x,y) = \frac{1}{2\pi} \int_{-n}^n e^{-i\omega(x-y)}\, d\omega
\end{equation*}
and thus
\begin{equation*}
  \tr[ P_{[-n_\delta(N),n_\delta(N)]} K_0^m ] = \frac{1}{2\pi} 2n_\delta(N) (2\pi)^{\frac{m}{2}} \int_\R \hat K_0(\omega)^m\, d\omega .
\end{equation*}
This implies
\begin{equation*}
  \tr[ P_{[-n_\delta(N),n_\delta(N)]} K_0^m ]
     = 2n_\delta(N) (2\pi)^{\frac{m-2}{2}} \int_\R \biggl[ \sqrt{\frac{\pi}{2}} \frac{1}{\cosh(\frac{\pi\omega}{2})} \biggr]^m\, d\omega
     = 2n_\delta(N)\pi^{m-2} \int_\R \frac{1}{[\cosh(\omega)]^m}\, d\omega
\end{equation*}
which proves the lemma.
\end{proof}

In order to bound the traces of the projection operator in \eqref{lc02t01}
we need pointwise estimates for the Laguerre polynomials.
The first one is Szeg\H{o}'s inequality, \cite[(7.21.3)]{Szego1975}, 
\begin{equation}\label{szegoe_inequality}
  |L_n(x)| \leq e^{\frac{x}{2}},\ x\geq 0,\ n\in\N_0.
\end{equation}
The second one is the less known Lewandowski--Szynal inequality \cite[Corollary 1]{LewandowskiSzynal1998}, which bounds
the Laguerre polynomial via the incomplete Gamma function
\begin{equation}\label{lewandowski_szynal_inequality}
  |L_n(x)| \leq \frac{e^x}{n!} \int_x^\infty t^n e^{-t}\, dt ,\ x\geq 0,\ n\in\N_0 .
\end{equation}
We will also need the simple formula
\begin{equation}\label{partial_sum}
  \sum_{k=0}^n \frac{1}{k!} x^k = \frac{e^x}{n!} \int_x^\infty t^n e^{-t}\, dt,
\end{equation}
whereby one could replace the integral in \eqref{lewandowski_szynal_inequality} by the partial sum of the exponential
function $e^x$. In particular, \eqref{lewandowski_szynal_inequality} is better for large $x$ than \eqref{szegoe_inequality}
but does not converge to \eqref{szegoe_inequality} for large $n$ and fixed $x$ because of the different exponents.

\begin{lemma}\label{lc04t}
Let $\delta\geq 0$. Furthermore, let $N\in\N$ and $L>0$ such that $\frac{N}{L}<\frac{1}{2}$. Then,
\begin{equation*}
    \tr[P_{[\delta,L]}^\perp \Pi_N ] \leq \delta (N+1) + \frac{4}{\frac{1}{2}-\frac{N}{L}} \frac{1}{N!} e^{-\frac{L}{2}} L^N  .
\end{equation*}
\end{lemma}
\begin{proof}
First note that $P_{[\delta,L]}^\perp = P_{[0,\delta]} + P_{\interval[open right]{L}{\infty}}$.
Using Szeg\H{o}'s inequality \eqref{szegoe_inequality} we obtain
\begin{equation*}
  \tr[P_{[0,\delta]} \Pi_N ]  = \int_0^\delta \sum_{n=0}^N l_n(x)^2 \, dx \leq \delta (N+1) .
\end{equation*}
The remaining trace is a bit more difficult. To simplify the calculations, we apply Szeg\H{o}'s inequality to one factor in
\begin{equation*}
  0 \leq \Pi_N(x,x) = \sum_{n=0}^N l_n(x)^2 \leq \sum_{n=0}^N |l_n(x)|,\ x\geq 0
\end{equation*}
and then use the Lewandowski--Szynal inequality \eqref{lewandowski_szynal_inequality}, $x\geq 0$,
\begin{equation*}
  0 \leq \Pi_N(x,x) 
    \leq \sum_{n=0}^N \frac{ e^{\frac{x}{2}}}{n!} \int_x^\infty t^n e^{-t}\, dt
    =  e^{\frac{x}{2}} \int_x^\infty e^{-t} \sum_{n=0}^N \frac{t^n}{n!} \, dt
    =  \frac{1}{N!} e^{\frac{x}{2}} \int_x^\infty \int_t^\infty s^N e^{-s}\, ds\, dt .
\end{equation*}
In the last step we used \eqref{partial_sum}. Furthermore,
\begin{equation*}
\begin{split}
  N! \tr[P_{\interval[open right]{L}{\infty}} \Pi_N]
    & = \int_L^\infty e^{\frac{x}{2}} \int_x^\infty s^N e^{-s}(s-x)\, ds\, dx\\
    & = e^{-\frac{L}{2}} \int_0^\infty e^{\frac{x}{2}} \int_x^\infty (s+L)^N e^{-s} (s-x)\, ds\, dx\\
    & = e^{-\frac{L}{2}} L^N \int_0^\infty (1+\frac{s}{L})^N e^{-\frac{s}{2}} \int_0^s e^{-\frac{x}{2}} x\, dx\, ds\\
    & \leq e^{-\frac{L}{2}} L^N \int_0^\infty e^{\frac{N}{L}s} e^{-\frac{s}{2}} \int_0^s e^{-\frac{x}{2}} x\, dx\, ds .
\end{split}
\end{equation*}
For simplicity we bound the $x$-integral by $4$
\begin{equation*}
  N! \tr[P_{\interval[open right]{L}{\infty}} \Pi_N]
   \leq 4 e^{-\frac{L}{2}} L^N \int_0^\infty e^{\frac{N}{L}s} e^{-\frac{s}{2}}\, ds
   =  4 e^{-\frac{L}{2}} L^N \frac{1}{\frac{1}{2}-\frac{N}{L}} .
\end{equation*}
This completes the proof.
\end{proof}

We combine the preceding estimates to obtain a bound on the trace of the Hilbert matrix.

\begin{lemma}\label{lc05t}
Let $\alpha\geq\frac{1}{2}$ and $N,m\in\N$ with $m\geq 5$. Then,
\begin{equation}\label{lc05t01}
  \frac{1}{\pi^{2^m}}\tr( H_{N,\alpha}^{2^m} ) 
    \leq C\Bigl\{ \frac{1}{2^{\frac{m}{2}}} \bigl[ \ln(m) + \ln(N) \bigr]
             +  \frac{1}{m^2} + \frac{1}{(2N)!}(mN)^{2N} e^{-\frac{1}{2}mN} \Bigr\}
\end{equation}
with some explicitely given constant $0\leq C<\infty$.
\end{lemma}
\begin{proof}
Lemmas \ref{lc01t} and \ref{lc02t} imply
\begin{equation}\label{lc05t02}
  \frac{1}{\pi^{2^m}} \tr[ H_{N,\alpha}^{2^m} ]
     \leq \frac{2}{\pi^{2^m}} \tr(P_{[\delta,L]} K^{2^m}) + 2 \tr\bigl( P_{[\delta,L]}^\perp\Pi_{2N} \bigr) .
\end{equation}
We let $\delta$ and $L$ depend on $m$ and $N$ in an appropriate way
\begin{equation*}
  \delta \coloneqq \frac{1}{(2N+1)m^2},\ L\coloneqq mN,
\end{equation*}
and bound the first term in \eqref{lc05t02} with the aid of Lemma \ref{lc03t} and \eqref{int01t02}
\begin{equation}\label{lc05t03}
  \frac{2}{\pi^{2^m}} \tr(P_{[\delta,L]} K^{2^m})
    \leq \frac{4}{\pi^2} \frac{n_\delta(L-\delta)}{\sqrt{2^{m-1}-1}}
    \leq \frac{2}{\pi^2} \frac{1}{2^{\frac{m}{2}}} \ln(m^3 N(2N+1)) .
\end{equation}
For the second term follows via Lemma \ref{lc04t} ($m\geq 5$)
\begin{equation}\label{lc05t04}
\begin{split}
  2 \tr\bigl( P_{[\delta,L]}^\perp\Pi_{2N} \bigr)
    & \leq 2 \Bigl( \delta(2N+1) + \frac{4}{\frac{1}{2}-\frac{2N}{L}} \frac{1}{(2N)!} L^{2N} e^{-\frac{L}{2}} \Bigr)\\
    & = 2\Bigl( \frac{1}{m^2} + \frac{4}{\frac{1}{2}-\frac{2}{m}} \frac{1}{(2N)!}(mL)^{2N} e^{-\frac{1}{2}mN} \Bigr) .
\end{split}
\end{equation}
Via some elementary estimates, \eqref{lc05t03} and \eqref{lc05t04} imply \eqref{lc05t01}.
\end{proof}

Now, everything is at hand to prove the complement of Proposition \ref{lc00t}.

\begin{proposition}\label{lc06t}
Let $\alpha\geq \frac{1}{2}$. Then,
\begin{equation}\label{lc06t01}
  -\liminf_{N\to\infty} \frac{1}{2n_{\frac{\alpha}{2}}(N)} \ln\bigl( \det(\id - \frac{1}{\pi} H_{N,\alpha} \bigr) \leq \frac{3}{4},\
   n_{\frac{\alpha}{2}}(N)=\frac{1}{4}\ln\Bigl( \frac{N+\frac{\alpha}{2}}{\frac{\alpha}{2}} \Bigr) .
\end{equation}
\end{proposition}
\begin{proof}
Since $H_{N,\alpha}^{2^m}$ is a non-negative operator the trace norm in \eqref{det01t02} equals the trace
\begin{equation}\label{lc06t02}
  -\ln(\det(\id-\frac{1}{\pi}H_{N,\alpha}))
     \leq \sum_{m=0}^M \ln\bigl(\det(\id + (\frac{1}{\pi}H_{N,\alpha})^{2^m})\bigr)
            +  \sum_{m=M+1}^\infty \frac{1}{\pi^{2^m}} \tr( H_{N,\alpha}^{2^m}) .
\end{equation}
We bound the traces via Lemma \ref{lc05t} (with $M\geq 4$)
\begin{equation*}
  \sum_{m=M+1}^\infty \frac{1}{\pi^{2^m}} \tr( H_{N,\alpha}^{2^m})
    \leq C_1 \sum_{m=M+1}^\infty \Bigl\{ \frac{1}{2^{\frac{m}{2}}} \bigl[ \ln(m) + \ln(N) \bigr] + \frac{1}{m^2} \Bigr\} + 
         C_1 \frac{N^{2N}}{(2N)!}\sum_{m=M+1}^\infty m^{2N} e^{-\frac{1}{2}mN} .
\end{equation*}
For the first sum
\begin{equation}\label{lc06t03}
  \liminf_{N\to\infty} \frac{1}{\ln(N)} 
    \sum_{m=M+1}^\infty \Bigl\{ \frac{1}{2^{\frac{m}{2}}} \bigl[ \ln(m) + \ln(N) \bigr] + \frac{1}{m^2} \Bigr\}
    = \sum_{m=M+1}^\infty \frac{1}{2^{\frac{m}{2}}} .
\end{equation}
The second series requires a bit more reasoning. For sufficiently large $M\in\N$,
\begin{equation*}
\begin{split}
  \frac{1}{(2N)!}N^{2N} \sum_{m=M+1}^\infty m^{2N}e^{-\frac{1}{2}mN}
   & \leq \frac{1}{(2N)!}N^{2N}\int_M^\infty t^{2N}e^{-\frac{1}{2}Nt}\, dt \\
   & =  \frac{1}{(2N)!}\frac{2}{N} e^{-\frac{1}{2}MN}(MN)^{2N}   \int_0^\infty\bigl( 1 + \frac{2t}{MN}\bigr)^{2N}e^{-t}\, dt\\
   & \leq \frac{1}{(2N)!}\frac{2}{N} e^{-\frac{1}{2}MN}(MN)^{2N}   \int_0^\infty e^{\frac{4t}{M}} e^{-t}\, dt\\
   & \leq C_2 \frac{1}{N^{\frac{3}{2}}} \bigl(\frac{e}{2}\bigr)^{2N} e^{-\frac{1}{2}MN} M^{2N} \\
   & \leq C_2 \frac{1}{N^{\frac{3}{2}}} \exp\bigl[ (2-2\ln(2)-\frac{1}{2}M + 2\ln(M))N \bigr]
\end{split}
\end{equation*}
with some constant $C_2\geq 0$.
In the next to last step we used the lower bound from Stirling's formula. For $M$ large enough, the argument of the exponential function becomes negative which
shows that the expression coverges to zero as $N\to\infty$ even without the factor $n_{\frac{\alpha}{2}}(N)$. 
Now, divide \eqref{lc06t02} by $2n_{\frac{\alpha}{2}}(N)$ and use Corollary \ref{slt06t} and \eqref{lc06t03} to deduce
\begin{equation}\label{lc06t04}
  -\liminf_{N\to\infty} \frac{1}{2n_{\frac{\alpha}{2}}(N)}\ln(\det(\id-\frac{1}{\pi}H_{N,\alpha})) 
    \leq  \sum_{m=0}^M \gamma_{2^m} + C_3 \sum_{m=M+1}^\infty \frac{1}{2^{\frac{m}{2}}}
\end{equation}
with $C_3\geq 0$ to adjust for $\ln(N)$ in \eqref{lc06t03} instead of $n_{\frac{\alpha}{2}}(N)$.
Since \eqref{lc06t04} is true for all (sufficiently large) $M\in\N$ we may perform the limit $M\to\infty$
\begin{equation*}
  -\liminf_{N\to\infty} \frac{1}{2n_{\frac{\alpha}{2}}(N)}\ln\bigl(\det(\id-\frac{1}{\pi}H_{N,\alpha})\bigr)
     \leq \sum_{m=0}^\infty \gamma_{2^m}  .
\end{equation*} 
We evaluate the infinite sum by using the explicit form of the $\gamma_k$'s in \eqref{slt06t02}
\begin{equation*}
  \sum_{m=0}^\infty \gamma_{2^m}
    = \frac{2}{\pi^2} \sum_{m=0}^\infty \int_0^\infty \ln\bigl(1+\frac{1}{[\cosh(\omega)]^{2^m}}\bigr) \, d\omega
    = \frac{2}{\pi^2} \int_0^\infty \ln\Bigl( \prod_{m=0}^\infty \bigl(1+\frac{1}{[\cosh(\omega)]^{2^m}}\bigr)\Bigr)\, d\omega .
\end{equation*}
Interchanging summation and integration can be justified via Lebesgue's convergence theorem. With \eqref{product} we obtain
\begin{equation*}
  \sum_{m=0}^\infty \gamma_{2^m} 
    = \frac{2}{\pi^2} \int_0^\infty \ln\Bigl( \frac{1}{1-\frac{1}{\cosh(\omega)}}\Bigr)\, d\omega
    = -\frac{2}{\pi^2} \int_0^\infty \ln\Bigl( 1-\frac{1}{\cosh(\omega)}\Bigr)\, d\omega
    = \frac{3}{4} .
\end{equation*}
In the last step we used Lemma \ref{int03t}. This yields \eqref{lc06t01}.
\end{proof}

We combine the lower and upper bound.

\begin{theorem}\label{lc07t}
Let $\alpha\geq\frac{1}{2}$. Then,
\begin{equation*}
  \ln\bigl(\det(\id - \frac{1}{\pi}H_{N,\alpha})\bigr)
     = 2n_{\frac{\alpha}{2}}(N)  \gamma(1) + o(\ln(N)) \ \text{as}\ N\to\infty,\
        n_{\frac{\alpha}{2}}(N)=\frac{1}{4}\ln\Bigl( \frac{N+\frac{\alpha}{2}}{\frac{\alpha}{2}} \Bigr)
\end{equation*}
with $\gamma(1)=-\frac{4}{3}$.
\end{theorem}
\begin{proof}
From Propositions \ref{lc00t} and \ref{lc06t} we obtain
\begin{equation*}
  -\frac{3}{4} 
    \leq \liminf_{N\to\infty} \frac{1}{2n_{\frac{\alpha}{2}}(N)} \ln\bigl( \det(\id - \frac{1}{\pi} H_{N,\alpha} \bigr) 
     \leq \limsup_{N\to\infty} \frac{1}{2n_{\frac{\alpha}{2}}(N)} \ln\bigl( \det(\id - \frac{1}{\pi}H_{N,\alpha}) \bigr)
     \leq \gamma(1)
      = -\frac{3}{4}, 
\end{equation*}
cf. \eqref{slt05t02}. This proves the statement.
\end{proof}

\section{\texorpdfstring{Limit case for $\alpha=1$}{Limit case for alpha=1}}
For the special Hilbert matrix with $\alpha=1$, cf. \eqref{hilbert_matrix},
there is an alternative way to prove the trace estimates (Lemmas \ref{lc01t}, \ref{lc02t}, \ref{lc04t}, \ref{lc05t})
used in Proposition \ref{lc07t} to bound the limit inferior.
Starting point is a simple estimate for the hyperbolic sine.

\begin{lemma}\label{lca01t}
Let $0\leq \delta\leq \frac{1}{3}$. Then, the hyperbolic sine satisfies the estimate
\begin{equation*}
  \frac{y}{\sinh(y)} \leq 2^\delta e^{-\delta y},\,\ y>0.
\end{equation*}
\end{lemma}
\begin{proof}
The left-hand side follows from $y\mapsto e^yy/\sinh(y)$ being an increasing function. For the right-hand side we use Lazarevic's
inequality \cite[3.6.9]{Mitrinovic1970}
\begin{equation*}
  \cosh(y) \leq \biggl[ \frac{\sinh(y)}{y}\biggr]^p,\ y\neq 0,\ p\geq 3.
\end{equation*}
For the proof note that $\sinh(y)/y\geq 1$ whence one only has to consider the case $p=3$. 
Using $\cosh(y)\geq e^y/2$ yields the claimed inequality with $\delta=1/p$.
\end{proof}

We replace the Hilbert matrix by the Carleman operator.

\begin{lemma}\label{lca02t}
Let $N,m\in\N$ and $0<\delta\leq \frac{1}{3}$. Then,
\begin{equation*}
  0\leq\tr[ H_{N,1}^m ] \leq 2^{m\delta} \tr[ P_{[\delta,N+\delta]}K^m ]
\end{equation*}
with $K$ the Carleman operator \eqref{carleman_operator}.
\end{lemma}
\begin{proof}
From Lemma \ref{hm01t} follows
\begin{equation*}
  \tr[ H_{N,1}^m ] = \tr[ G_{N,1}^m ],\ m\in\N .
\end{equation*}
Recall the kernel function (see Lemma \ref{hm01t} and the proof of Lemma \ref{hm02t})
\begin{equation*}
  G_{N,1}(x) = \frac{x}{2\sinh(\frac{x}{2})} \int_0^N e^{-sx}\, ds .
\end{equation*}
With the aid of Lemma \ref{lca01t}
\begin{equation*}
  0 \leq G_{N,1}(x+y)
    \leq 2^{\delta}e^{-\delta(x+y)} \int_0^N e^{-s(x+y)}\, ds
    = 2^\delta \int_0^N e^{-(s+\delta)(x+y)}\, ds
    = 2^\delta (E_{2\delta}P_{[0,N]}E_{2\delta}^*)(x,y)
\end{equation*}
where $E_{2\delta}$ is from \eqref{hm02t01} with $\alpha=2\delta$. Since $\delta>0$ we may take the trace, Lemma \ref{hm02t}
\begin{equation*}
  0 \leq \tr[ H_{N,1}^m ] = \tr[ G_{N,1}^m] \leq 2^{m\delta} \tr[ (E_{2\delta}P_{[0,N]}E_{2\delta}^* )^m ]
\end{equation*}
where we used that the kernel functions are (pointwise) non-negative. Via Lemma \ref{hm03t}
\begin{equation*}
  \tr[ (E_{2\delta}P_{[0,N]}E_{2\delta}^*)^m ] 
   = \tr[ (P_{[\delta,N+\delta]}K P_{[\delta,N+\delta]} )^m ] 
   \leq \tr[ P_{[\delta,N+\delta]} K^m ]
\end{equation*}
In the last step we used $0\leq P_{[\delta,N+\delta]}\leq \id$ in the sense of quadratic forms.
\end{proof}

We replace the Carleman operator $K$ by the convolution operator $K_0$.

\begin{lemma}\label{lca03t}
Let $0<\delta\leq \frac{1}{3}$. With the convolution operator $K_0$ from Lemma \ref{slt02t}
\begin{equation*}
  \tr[ P_{[\delta,N+\delta]} K^m ] = \tr[ P_{[-n,n]} K_0^m ],\  n_\delta(N)=\frac{1}{4}\ln\frac{N+\delta}{\delta} .
\end{equation*}
\end{lemma}
\begin{proof}
See Lemma \ref{slt02t}.
\end{proof}

Using the diagonalization of the convolution operator $K_0$, see \eqref{symbol}, we express the trace as a simple integral.

\begin{lemma}\label{lca04t}
Let $m\in\N$ and $n\geq 0$. Then,
\begin{equation*}
  \tr[ P_{[-n,n]} K_0^m ] = 2n\pi^{m-2} \int_\R \frac{1}{[\cosh(\omega)]^m}\, d\omega .
\end{equation*}
\end{lemma}
\begin{proof}
Via the diagonalization $\mathcal{F}K_0\mathcal{F}^* = \sqrt{2\pi} \hat K_0$, see \eqref{fourier_transform} and \eqref{symbol},
we obtain
\begin{equation*}
  \tr[ P_{[-n,n]} K_0^m ]
    = (2\pi)^{\frac{m}{2}} \tr[ P_{-n,n]} \mathcal{F}^* \hat K_0^m \mathcal{F} ]
    = (2\pi)^{\frac{m}{2}} \tr[ \mathcal{F} P_{[-n,n]} \mathcal{F}^* \hat K_0^m ] .
\end{equation*}
Now,
\begin{equation*}
  \mathcal{F} P_{[-n,n]} \mathcal{F}^*(x,y) = \frac{1}{2\pi} \int_{-n}^n e^{-i\omega(x-y)}\, d\omega
\end{equation*}
and thus
\begin{equation*}
  \tr[ P_{[-n,n]} K_0^m ] = \frac{1}{2\pi} 2n (2\pi)^{\frac{m}{2}} \int_\R \hat K_0(\omega)^m\, d\omega .
\end{equation*}
This implies
\begin{equation*}
  \tr[ P_{[-n,n]} K_0^m ]
     = 2n (2\pi)^{\frac{m-2}{2}} \int_\R \biggl[ \sqrt{\frac{\pi}{2}} \frac{1}{\cosh(\frac{\pi\omega}{2})} \biggr]^m\, d\omega
     = 2n\pi^{m-2} \int_\R \frac{1}{[\cosh(\omega)]^m}\, d\omega
\end{equation*}
which proves the lemma.
\end{proof}

We give now a new proof of Proposition \ref{lc06t}. We formulate only the relevant part.

\begin{proposition}\label{lca05t}
The special Hilbert matrix $H_{N,1}$, cf. \eqref{hilbert_matrix}, satisfies
\begin{equation*}
  -\liminf_{N\to\infty} \frac{1}{2n_{\frac{1}{2}}(N)}\ln(\det(\id-\frac{1}{\pi}H_{N,1})) 
    \leq  \sum_{m=0}^\infty \gamma_{2^m},\
       n_{\frac{1}{2}}(N)=\frac{1}{4}\ln\Bigl( \frac{N+\frac{1}{2}}{\frac{1}{2}} \Bigr) .
\end{equation*}
\end{proposition}
\begin{proof}
We start from \eqref{lc06t02} but use now Lemmas \ref{lca02t} through \ref{lca04t}. These imply (we only need even exponents)
\begin{equation*}
  \frac{1}{\pi^{2k}} \tr[ H_{N,1}^{2k} ]
    \leq \frac{2n_\delta(N)}{\pi^2} 2^{2k\delta} \int_\R \frac{1}{[\cosh(\omega)]^{2k}} \, d\omega,\
      0< \delta \leq \frac{1}{3},\ n_\delta(N) = \frac{1}{4}\ln\frac{N+\delta}{\delta} ,
\end{equation*}
which can be further estimated with the aid of \eqref{int01t02}
\begin{equation*}
  \frac{1}{\pi^{2k}} \tr[ H_{N,1}^{2k} ] \leq \frac{2n_\delta(N)}{\pi^2} 2^{2k\delta} \frac{2}{\sqrt{k-1}},\ k\geq 2 .
\end{equation*}
In order to compensate the exponentially growing prefactor we choose $\delta=\frac{1}{k}$,
\begin{equation*}
  \frac{1}{\pi^{2k}} \tr[ H_{N,1}^{2k} ] \leq \frac{16}{\pi^2} \frac{n_{\frac{1}{k}}(N)}{\sqrt{k-1}},\ n_{\frac{1}{k}}(N) = \frac{1}{4}\ln[ ( N+\frac{1}{k})k].
\end{equation*}
Now we can estimate the infinite sum in \eqref{lc06t02}
\begin{equation*}
\begin{split}
  \sum_{m=M+1}^\infty \frac{1}{\pi^{2^m}} \tr[ H_{N,1}^{2^m} ]
     & \leq \frac{16}{\pi^2} \sum_{m=M+1}^\infty \frac{1}{\sqrt{2^{m-1}-1}} \frac{1}{4} \ln\bigl( ( N+\frac{1}{2^{m-1}} ) 2^{m-1} \bigr)\\
     & \leq \frac{4}{\pi^2} \sum_{m=M}^\infty \frac{1}{\sqrt{2^m-1}}\bigl\{ m\ln(2) + \ln(N+\frac{1}{2^m}) \bigr\} \\
     & \leq C_1 \sum_{m=M}^\infty \frac{m}{\sqrt{2^m-1}} + C_2 \ln(N+1) \sum_{m=M}^\infty \frac{1}{\sqrt{2^m-1}} .
\end{split}
\end{equation*}
This yields the analogue of \eqref{lc06t04}
\begin{equation*}
  -\liminf_{N\to\infty} \frac{1}{2n_{\frac{1}{2}}(N)}\ln(\det(\id-\frac{1}{\pi}H_{N,1})) 
    \leq  \sum_{m=0}^M \gamma_{2^m} + C_3 \sum_{m=M}^\infty \frac{1}{\sqrt{2^m-1}} .
\end{equation*}
Letting $M\to\infty $ we obtain the statement.
\end{proof}

\appendix
\section{Integrals\label{int}}
\begin{lemma}\label{int01t}
Let $m\in\N$. Then,
\begin{equation}\label{int01t01}
  I_{2m} 
   \coloneqq \int_\R \frac{1}{ \cosh(x)^{2m} }\, dx 
    = 2 \prod_{k=1}^{m-1} \frac{2k}{2k+1} 
    = 2\frac{4^{m-1}[(m-1)!]^2}{(2m-1)!},
\end{equation}
which can be estimated
\begin{equation}\label{int01t02}
  I_{2m+2} \leq \frac{2}{\sqrt{m}} ,\ m\in\N.
\end{equation} 
\end{lemma}
\begin{proof}
We note $\frac{d}{dx} \tanh(x) = 1/\cosh(x)^2$ and integrate by parts
\begin{equation*}
\begin{split}
  I_{2m+2}
   & = \int_\R \frac{1}{[\cosh(x)]^{2m}} \frac{1}{[\cosh(x)]^2}\, dx\\
   & = \biggl[ \frac{1}{[\cosh(x)]^{2m}} \frac{\sinh(x)}{\cosh(x)} \biggr]_{-\infty}^\infty
         + 2m \int_\R \frac{\sinh(x)}{[\cosh(x)]^{2m+1}} \frac{\sinh(x)}{\cosh(x)}\, dx\\
   & = 2m\int_\R \frac{[\cosh(x)]^2}{[\cosh(x)]^{2m+2}}\, dx - 2m \int_\R \frac{1}{[\cosh(x)]^{2m+2}}\, dx\\
   & = 2m I_{2m} - 2m I_{2m+2} .
\end{split}
\end{equation*}
We solve for $I_{2m+2}$ to obtain the recursion formula
\begin{equation*}
  I_{2(m+1)} = \frac{2m}{2m+1} I_{2m}
\end{equation*}
which immediately yields
\begin{equation*}
  I_{2(m+1)} = 2 \prod_{k=1}^m \frac{2k}{2k+1} = 2\prod_{k=1}^m \frac{k}{k+\frac{1}{2}} 
\end{equation*}
since $I_2=2$. This implies \eqref{int01t01}. In order to derive the bound we use the inequality between
the geometric and arithmetic mean
\begin{equation*}
  I_{2(m+1)} = 2 \frac{\sqrt{m}}{m+\frac{1}{2}} \frac{\sqrt{m}\sqrt{m-1}}{m-\frac{1}{2}} \frac{\sqrt{m-1}\sqrt{m-2}}{m-\frac{3}{2}}
               \cdots \frac{\sqrt{2}\sqrt{1}}{1 + \frac{1}{2}} \sqrt{1}
           \leq 2 \frac{\sqrt{m}}{m+\frac{1}{2}} 
           \leq \frac{2}{\sqrt{m}} .
\end{equation*}
This proves \eqref{int01t02}.
\end{proof}

The following integral is a special case of an integral that appeared in the study of the ground state energy of the free Fermi
gas \cite{OtteSpitzer2018}. We evaluate it here for the sake of completeness.

\begin{lemma}\label{int02t}
Let $\beta\in\C\setminus\interval[open right]{1}{\infty}$. Then,
\begin{equation}\label{int02t01}
  I(\beta) \coloneqq \int_0^\infty \ln\bigl( 1 - \frac{\beta}{\cosh(x)} \bigr)\, dx
     = \frac{1}{2} [\arcosh(-\beta)]^2 + \frac{\pi^2}{8} .
\end{equation}
\end{lemma} 
\begin{proof}
First of all, we transform the integral into a form that can be treated by standard methods.
To this end, we write $f(x) = \cosh(x)-1$ for short. Note that $f(0)=0$, $f(\infty)=\infty$, and $f'(x)>0$ for $x>0$. Therefore,
\begin{equation*}
  x = f^{-1}(y),\ dx = \frac{d}{dy} (f^{-1}(y))\, dy,
\end{equation*}
is a well-defined substitution. Hence,
\begin{equation*}
  I(\beta) = \int_0^\infty \ln( 1 - \frac{\beta}{f(x)+1})\, dx
           = \int_0^\infty \ln(1 - \frac{\beta}{y+1} )\frac{d}{dy}(f^{-1}(y))\, dy .
\end{equation*}
An integration by parts yields
\begin{equation*}
  I(\beta) = -\beta \int_0^\infty \frac{1}{y+1-\beta} \frac{1}{y+1} f^{-1}(y)\, dy
           = \int_0^\infty \bigl[ \frac{1}{y+1} - \frac{1}{y+1-\beta} \bigr] f^{-1}(y)\, dy .
\end{equation*}
The integral is of the type
\begin{equation*}
  I(\beta) = \int_0^\infty r(y) g(y) \, dy,\ g(y)\coloneqq f^{-1}(y)
\end{equation*}
where the rational function $r$ does not have poles in $\interval[open right]{0}{\infty}$. Such integrals
can be evaluated by standard methods if one finds a function $h$ with a certain jump at $\interval[open right]{0}{\infty}$. In our case
\begin{equation*}
  h(z) \coloneqq -\frac{1}{4\pi i} [\arcosh(-z-1)]^2 .
\end{equation*}
Then, via the residue theorem
\begin{equation*}
\begin{split}
  I(\beta) & = 2\pi i \sum_{z\in\C\setminus\interval[open right]{0}{\infty}} \res(r(z)h(z)) \\
           & = \frac{1}{2} \sum_{z\in\C\setminus\interval[open right]{0}{\infty}} \res\biggl[ \frac{1}{z+1-\beta}[\arcosh(-z-1)]^2\biggr]
                 - \frac{1}{2} \sum_{z\in\C\setminus\interval[open right]{0}{\infty}} \res\biggl[ \frac{1}{z+1}[\arcosh(-z-1)]^2\biggr]\\
           & = \frac{1}{2}[\arcosh(-\beta)]^2 - \frac{1}{2}[\arcosh(0)]^2 
\end{split}
\end{equation*}
which yields \eqref{int02t01}.
\end{proof}

The method used to prove the preceding lemma does not work in the case $\beta=1$. One could use a continuity argument
to cover this case as well. Instead, we transform the integral into a well-known integral.

\begin{lemma}\label{int03t}
Let $\beta=1$ in Lemma \ref{int02t}. Then,
\begin{equation}\label{int03t01}
  I(1) = \int_0^\infty \ln\bigl(1-\frac{1}{\cosh(x)}\bigr)\, dx = -\frac{3\pi^2}{8} .
\end{equation}
\end{lemma}
\begin{proof}
Despite the singularity at $x=0$ the integral is well-defined since the logarithm $x\mapsto \ln(x)$ is integrable a $x=0$.
We integrate by parts and use some standard formulae for the hyperbolic functions
\begin{equation*}
\begin{split}
  I(1) 
   & = - \int_0^\infty \frac{x}{\cosh(x)-1} \frac{\sinh(x)}{\cosh(x)}\, dx\\
   & = - \int_0^\infty \frac{x}{[\cosh(x)]^2-1} \frac{(\cosh(x)+1)\sinh(x)}{\cosh(x)}\, dx\\
   & = - \int_0^\infty \frac{x}{\sinh(x)}\, dx - \int_0^\infty \frac{x}{\sinh(x)\cosh(x)}\, dx\\
   & = - \frac{3}{2}  \int_0^\infty \frac{x}{\sinh(x)}\, dx .
\end{split}
\end{equation*}
The latter integral is well-known and has the value $\frac{\pi^2}{4}$.
It can be evaluated via Cauchy's integral theorem and an appropriate integration
contour. A possible choice is the rectangle with vertices $\pm R$ and $\pm R+i\pi$ with a small half circle at $i\pi$ cut out. 
\end{proof}

%
%%%%%%%%%%%%%%%%%%%%%%%%%%%%%%%%%%%%%%%%%%%%%%%%%%%%%%%%%%%%%%%%%%%%%%%%%%%%%%
%
%\section*{Acknowledgement}
%\input{acknowledgement}
%
\bibliographystyle{plain}
\bibliography{literature} % dummy file
\end{document}